\documentclass[final]{amsart}
 
\usepackage{microtype}

\title{A Monadic Framework for Interactive Realizability}
\author[1]{Giovanni Birolo}

\usepackage{mdframed} 
\usepackage{enumitem} 
\usepackage{cleveref} 
\usepackage{amsthm} 

\usepackage{txfonts}
\usepackage{ifthen}

\newcommand\phantomrel[1]{\mathrel{\phantom{#1}}} 
\newcommand\reducesto\to

\newcommand\N{\mathbb{N}}

\newcommand\num[1]{\boldsymbol{#1}} 
\newcommand\tarrow\to 

\newcommand\seq\vdash 
\newcommand\monSeq\Vdash 
\newcommand\mon\mathfrak 

\newcommand\ifNotEmpty[2]{\ifthenelse{\equal{#1}{}}{}{#2}}
\newcommand\optionalSubscript[1]{\ifthenelse{\equal{#1}{}}{}{_{#1}}}
\newcommand\optionalSuperscript[1]{\ifthenelse{\equal{#1}{}}{}{^{#1}}}

\newcommand\MonadStyle\mathfrak

\newcommand\Unit{\mathrm{Unit}} 
\newcommand\Bool{\mathrm{Bool}} 
\newcommand\Nat{\mathrm{Nat}} 
 
\newcommand\State{\mathrm{State}} 
\newcommand\Ex{\mathrm{Ex}} 

\newcommand\abstr[2]{\lambda #1^{#2}.}

\newcommand\abstrS[1]{\abstr{#1}\State}
\DeclareMathOperator\unit\ast

\newcommand\newtermconstant[3]{\newcommand{#1}[1][]{#2
\ifthenelse{\equal{##1}{}}{}{^{##1}}%
\ifthenelse{\equal{#3}{}}{}{_\textrm{#3}}}%
}
\newcommand\newtermconstantN[2]{\newcommand{#1}[2][]{#2
\ifthenelse{\equal{##1}{}}{}{^{##1}}%
\ifthenelse{\equal{##2}{}}{}{_{##2}}}%
}

\newcommand\termname[1]{\operatorname{\textsf{\small#1}}}
\newtermconstant\pair{\termname{pair}}{}
\newtermconstant\prr{\termname{pr}}{R}
\newtermconstant\prl{\termname{pr}}{L}
\newtermconstant\inr{\termname{in}}{R}
\newtermconstant\inl{\termname{in}}{L}
\newtermconstant\case{\termname{case}}{}
\newtermconstantN\totRec{\termname{crec}}
\newtermconstantN\monStarN{\operatorname{\mon{star}}}
\newtermconstantN\monRaiseN{\operatorname{\mon{raise}}}
\newtermconstant\monBind{\operatorname{\mon{bind}}}{}
\newcommand\exmerge{\termname{merge}}
\newcommand\eval{\termname{eval}}
\newcommand\query{\termname{query}}

\newtermconstant\reg{\termname{reg}}{}
\newtermconstant\exc{\termname{ex}}{}
\newtermconstant\ite{\termname{ite}}{}
\newtermconstant\true{\termname{true}}{}
\newtermconstant\false{\termname{false}}{}
\newtermconstant\dummy{\termname{dummy}}{}
\newtermconstant\zero{\termname{zero}}{}
\let\succ\undefined
\newtermconstant\succ{\termname{succ}}{}

\newcommand\interp[1]{\llbracket #1 \rrbracket}

\usepackage{bussproofs}
\usepackage{txfonts}
\usepackage{ifthen}
\usepackage{tikz}
\usetikzlibrary{patterns}
\usetikzlibrary{decorations.pathmorphing}

\providecommand\optionalSubscript[1]{\ifthenelse{\equal{#1}{}}{}{_{#1}}}
\providecommand\optionalSuperscript[1]{\ifthenelse{\equal{#1}{}}{}{^{#1}}}

\newcommand\fa{A}
\newcommand\fb{B}
\newcommand\fc{C}

\newcommand\ltrue\top
\newcommand\lfalse\bot

\newcommand\limply\to 
\newcommand\subst[2]{[ #1 \coloneqq #2 ]}
\newcommand\quant[2]{#1 #2.\ }
\newcommand\qlambda[1]{\quant\lambda{#1}}
\newcommand\qforall{\quant\forall}
\newcommand\qexists{\quant\exists}

\tikzstyle{label} = [fill=white,inner sep=1pt]
\newcommand\IfNotEmpty[2]{\ifthenelse{\equal{#1}{}}{}{#2}}

\newcommand\PrDer[3]{
\begin{tikzpicture}[baseline=0.3cm,scale=1]
  \IfNotEmpty{#3}{ 
    \draw decorate [decoration={snake,amplitude=2pt,segment length=14pt}] {(0,0.2) --  node[above,label] {\(#3\)} (0,1.25)};
  }
  \draw (-0.5,1) -- (-0.1,0) -- (0.1,0) -- (0.5,1);
  \IfNotEmpty{#1}{ 
    \node[above,label] () at (0,0.2) {\(#1\)}; 
  } 
    
\end{tikzpicture}
}

\newcommand\PrAx[2][]{\AxiomC{\(#2\optionalSuperscript{#1}\)}}
\newcommand\PrUn[2][]{\UnaryInfC{\(#2\optionalSuperscript{#1}\)}}
\newcommand\PrBin[2][]{\BinaryInfC{\(#2\optionalSuperscript{#1}\)}}
\newcommand\PrTri[2][]{\TrinaryInfC{\(#2\optionalSuperscript{#1}\)}}
\newcommand\PrLbl[2][]{\LeftLabel{\(#2\)}\ifthenelse{\equal{#1}{}}{}{\RightLabel{\(#1\)}}}

\newcommand\PrInf[1][]{\ifthenelse{\equal{#1}{}}{%
\def\extraVskip{-2pt}\noLine\UnaryInfC\vdots\noLine\def\extraVskip{2pt}}{%
\noLine\UnaryInfC{\(#1\)}}}
\newcommand\RuleName[3][]{#2\mathrm{#3}\IfNotEmpty{#1}{_\mathrm{#1}}}
\newcommand\RuleNameI[2][]{\RuleName[#1]{#2}{I}}
\newcommand\RuleNameE[2][]{\RuleName[#1]{#2}{E}}

\usepackage{comment} 
\usepackage{ifdraft}
\usepackage{xcolor}

\newcommand\newmarkedenvironment[2]{%
  \newenvironment{#1}{\noindent\textcolor{red}{*** BEGIN: {#2} ***}\\}{\textcolor{red}{\\ *** END: {#2} ***}}
}

\ifoptiondraft{
  \newmarkedenvironment{wip}{Work in progress}
  \newmarkedenvironment{omitted}{Omitted}
}{
  \excludecomment{wip}
  \excludecomment{omitted}
}

\ifoptionfinal{
  \includecomment{changed}
  \newcommand\fixme[1]{}
  \newcommand\note[1]{}
}{
  \newmarkedenvironment{changed}{Changed from last version}
  \newcommand\fixme[1]{\textcolor{red}{ (X)}\marginpar{\textcolor{red}{#1}}}
  \newcommand\note[1]{\textcolor{green}{ (X)}\marginpar{\textcolor{green}{#1}}}
}

\theoremstyle{plain}
\newtheorem{theorem}{Theorem} 
\newtheorem{lemma}{Lemma} 
 
\newtheorem{definition}{Definition} 
\newtheorem{example}{Example} 
\newtheorem{remark}{Remark}

\crefname{property}{property}{properties}
\crefname{lstlisting}{program}{programs}

\newenvironment{legend}{\small \it}{}

\newcommand\HA{\mathsf{HA}}
\newcommand\EM{\mathsf{EM}_1}
\newcommand\PRA{\mathsf{PRA}}

\newcommand\lva{x}
\newcommand\lvb{y}
\newcommand\lvc{z}
\newcommand\lta{t}

\newcommand\lafa{P}

\newcommand\MM{\mon M}
\newcommand\IdM{\mon{Id}}
\newcommand\ExM{\mon{Ex}}
\newcommand\IR{\mon{IR}}

\def\monUnit{}
\def\monStar{}
\def\monMerge{}
\newcommand\setmonad[1]{%
  \def\mm##1{\lVert ##1 \rVert\ifthenelse{\equal{#1}{}}{}{_{#1}}}%
  \def\m##1{\lvert ##1 \rvert\ifthenelse{\equal{#1}{}}{}{_{#1}}}%
  \def\monTrans##1{\llbracket ##1 \rrbracket\optionalSubscript{#1}}%
  \def\monRe{\mathrel{\mon R}\ifthenelse{\equal{#1}{}}{}{_{#1}}}%
  \def\re{\mathrel{\mathtt R}\ifthenelse{\equal{#1}{}}{}{_{#1}}}%
  \def\T{T\optionalSubscript{#1}}%
  \renewcommand\monUnit[1][]{\operatorname{\mon{unit}}\optionalSubscript{#1}\optionalSuperscript{##1}}%
  \renewcommand\monStar[1][]{\operatorname{\mon{star}}\optionalSubscript{#1}\optionalSuperscript{##1}}%
  \renewcommand\monMerge[1][]{\operatorname{\mon{merge}}\optionalSubscript{#1}\optionalSuperscript{##1}}%
  \def\monSeq{\Vdash\optionalSubscript{#1}}
}

\begin{document}

\maketitle

\begin{abstract}
We give a new presentation of interactive realizability with a more explicit syntax. 

Interactive realizability is a realizability semantics that extends the Curry-Howard correspondence to (sub-)classical logic, more precisely to first-order intuitionistic arithmetic (Heyting Arithmetic) extended by the law of the excluded middle restricted to simply existential formulas, a system motivated by its applications in proof mining. 

Monads can be used to structure functional programs by providing a clean and modular way to include impure features in purely functional languages. 
We express interactive realizers by means of an abstract framework that applies the monadic approach used in functional programming to modified realizability, in order to obtain more ``relaxed'' realizability notions that are suitable to classical logic. In particular we use a combination of the state and exception monads in order to capture the learning-from-mistakes nature of interactive realizers at the syntactic level.
\end{abstract}

\section{Introduction}




\begin{omitted}
While constructive and in particular intuitionistic proofs have a computational content by design, this is not the case for proofs in classical logic. 
In general it is not possible to give a computational content to any classically provable statement: for instance, 
giving a computational interpretation of the classical law of the excluded middle that satisfies the BHK interpretation means being able to effectively decide for any statement whether it holds or not, which leads to contradictions in the case of many statements. 
On the other hand saying that classical logic has no computational content at all is unreasonable: 
for instance, given a classically provable statement we can say that its computational content is that of an intuitionistic proof of its double-negation translation. 
The problem with this approach is that the double-negation translation 
transforms informative statements into negative, non-informative ones.
Thus, relating the original statement with the computational content of its translation is not usually straightforward. 
\end{omitted}

The Curry-Howard correspondence was originally discovered for intuitionistic proofs. 
This is not coincidental: 
the type systems needed to interpret intuitionistic proofs are usually very simple and natural, as in the case of Heyting Arithmetic and System T (see \cite{girard88}). 
While classical proofs can be transformed into intuitionistic ones by means of the double-negation translation and then translated into typed programs, 
the existence of a direct correspondence was deemed unlikely until Griffin showed otherwise in \cite{griffin90}. 

Starting with Griffin's, other interpretations extending the Curry-Howard correspondence to classical logic have been put forward. 
Griffin uses a ``typed Idealized Scheme'' with the control construct {\sf call/cc}, that allows access to the current continuation. 
In \cite{parigot92}, Parigot introduces the \(\lambda\mu\)-calculus, an extension of lambda calculus with an additional kind of variables for subterms. 
In \cite{krivine94}, Krivine uses lambda calculus with a non-standard semantics, described by an abstract machine that allows the manipulation of ``stacks'', which can be thought of as execution contexts. 

All these different approaches seems to suggest that, in order to interpret classical logic, we need control operators or some syntactical equivalent thereof. 
This could be generalized in the idea that ``impure'' computational constructs are needed in order to interpret non-constructive proofs. 
Monads are a concept from category theory that has been widely used in computer science. 
In particular, they can be used to structure functional programs that mimic the effects of impure features. 

In \cite{moggi91}, Moggi advocates the use of monads as a framework to describe and study many different ``notions of computation'' in the context of categorical semantics of programming languages. 
A different take on the same idea that actually eschews category theory completely is suggested in \cite{wadler92} by Wadler: 
the definition of monad becomes purely syntactic and is used as a framework to structure functional programs by providing a clean and modular way to include impure features in purely functional languages (one noteworthy example is I/O in Haskell). 

The main idea of this work is to use monads as suggested by Wadler in order to structure programs extracted from classical proofs by interactive realizability. 
Recently introduced by Berardi and de'Liguoro in \cite{berardidL08,berardidL09}, 
interactive realizability is yet another technique for understanding and extracting the computational content in the case of the sub-classical logic \(\HA+\EM\) (Heyting Arithmetic extended by the law of the excluded middle restricted to \(\Sigma^0_1\) formulas).
Interactive realizability combines Coquand's game theoretic semantics for classical arithmetic \cite{coquand95} and Gold's idea of learning in the limit \cite{gold65}. 

A program extracted by means interactive realizability, called interactive realizer, can be thought of as a learning process. 
It accumulate information in a knowledge state and use this knowledge in order to ``decide'' the instances of \(\EM\) used in the proof. 
Since these instances are in general undecidable, the realizer actually makes an ``educated guess'' about which side of an \(\EM\) instance is true by looking at the state. 
Such guesses can be wrong. 

This can become apparent later in the proof, when the guessed side of the \(\EM\) instance is used to deduce some decidable statement. 
If this decidable statement turns out to be false, then the guess was wrong and the proof cannot be completed. 
In this case the realizer cannot produce the evidence required for the final statement and fails. 
However failure is due to the fact that we made a wrong guess. 
We can add this information to the state, so that, using this new state, we will be able to guess the \(\EM\) instance correctly. 
At this point we discard the computation that occurred after the wrong guess and we resume from there. 
This time we guess correctly and can proceed until the end or until we fail again because we guessed incorrectly another \(\EM\) instance. 

There are three ``impure'' parts in the behavior we described: the dependency on the knowledge state, the possibility of failure to produce the intended result and the backtracking after the failure.
In this work we use a monadic approach to describe the first two parts which are peculiar to interactive realizability. 
We do not describe the third part, which is common also to the other interpretations of classical logic. 

\begin{omitted}
Interactive realizability is based on the idea of learning by trial and error. 
More precisely: the computation of an interactive realizer depends on a finite knowledge state. 
If the knowledge state contains enough information, the realizer, computed on such state, will behave according to BHK semantics and produce the intended result. 
If the state is too small the realizer will fail and yield the missing piece of knowledge instead. 
Thus we have the following learning process: we begin by computing a realizer starting from the empty state. 
If the realizer succeeds we are done, otherwise we fail and add the piece of information obtained from the failure of the realizer to the state and compute the realizer on the updated state. 
We carry on this procedure until the realizer succeeds. 
The main computational properties of interactive realizers are thus exceptions (since they can fail) and dependency on a state. 
\end{omitted}


This paper is an account of interactive realizability where interactive realizers are encoded as \(\lambda\)-terms following the monadic approach to structuring functional programs suggested by Wadler. 
We shall prove that our presentation of interactive realizability is a sound semantics for \(\HA+\EM\). 


In our presentation, interactive realizer are written in a simply typed \(\lambda\)-calculus with products, coproducts and natural numbers with course-of-value recursion, extended with some abstract terms to represent states and exceptions. 
The peculiar features of interactive realizability, namely the dependency on the knowledge state and the possibility of failure, are explicitly computed by the \(\lambda\)-terms encoding the realizers. 
Thus the computational behavior of interactive realizers is evident at the syntactic level, without the need for special semantics. 

While proving the soundness of \(\HA+\EM\) with respect to our definition of interactive realizability, we observed that the soundness of \(\HA\) did not require any assumption on the specific monad we chose to structure interactive realizers (while the soundness of \(\EM\) requires them as expected). 
Prompted by this observation, we split the presentation in two parts. 

The former is an abstract monadic framework for producing realizability notions where the realizers are written in monadic style. 
We prove that \(\HA\) is sound with respect to any realizability semantics defined by the framework, for any monad. 
The latter is an application of this abstract framework to interactive realizability. 
We define the specific monad we use to structure interactive realizers and show that,  
by means of this specific monad, we can realize the \(\EM\) axiom.


This work builds on the presentation of interactive realizability given in \cite{aschieriB10} by Aschieri and Berardi. 
The main contributions with respect to \cite{aschieriB10} is a more precise description of the computational behavior of interactive realizer. 
This is explained in more detail at the end of the paper. 

Monads have first been used to describe interactive realizability by Berardi and de'Liguoro in \cite{berardidL10} and \cite{berardidL11}, where interactive realizers for \(\PRA+\EM\) are given a monadic categorical semantics following Moggi's approach. 
While our idea of using monad to describe interactive realizability was inspired by \cite{berardidL10}, our work is mostly unrelated: our use of monads follows Wadler's syntactical approach and we employ a different monad that emphasizes different aspects of interactive realizability.

\section{A Simply Typed $\lambda$-Calculus for Realizability}

\newcommand\ta{X}
\newcommand\tb{Y}
\newcommand\tc{Z}
\newcommand\tta{x} 
\newcommand\ttb{y} 
\newcommand\ttc{z} 
\newcommand\tva{x} 
\newcommand\tvb{y} 
\newcommand\tvc{z} 
\newcommand\ra{p}
\newcommand\rb{q}
\newcommand\rr{r}

In this section we introduce system \(T'\), a simply typed \(\lambda\)-calculus variant of G\"odel's system \(T\) in which we shall write our realizers.
System \(T'\) will be more convenient for our purposes in order to get a more straightforward translation of monads and related concepts from category theory. 
There are two main differences between our system \(T'\) and system \(T\). 
The first one is that we replace the boolean type with the more general sum (or co-product) type. 
The second one is that the recursion operator uses complete recursion instead of standard primitive recursion. 

We begin by defining the types. 
We shall use the metavariables \(\ta,\tb\) and \(\tc\) for types. 
We assume that we have a finite set of atomic types that includes the unit type \(\Unit\) and the type of natural numbers \(\Nat\). 
Moreover we have three type binary type constructors \(\tarrow, \times, +\). 
In other words, for any types \(\ta\) and \(\tb\) we have the arrow (or function) type \(\ta \tarrow \tb\), the product type \(\ta \times \tb\) and the sum (or co-product) type \(\ta + \tb\). 

We can now define the typed terms of the calculus. 
We assume that we have a countable set of typed term constants that includes the constructors and the destructors for the unit, natural, product and sum types (listed in \cref{fig:constant_terms}) and  
a countable set of variables of type \(\ta\) for any type \(\ta\): 
\[ \tva_0 : \ta, \dotsc, \tva_n : \ta, \dotsc. \]
\begin{figure}
  \caption{Constructors and destructors}\label{fig:constant_terms} 
  \begin{gather*} 
    \unit : \Unit, \\ 
    \pair[\ta,\tb] : \ta \tarrow \tb \tarrow \ta \times \tb, \\ 
    \begin{aligned} 
      \prl[\ta,\tb] &: \ta \times \tb \tarrow \ta, \qquad 
      &\prr[\ta,\tb] : \ta \times \tb \tarrow \tb, \\ 
      \inl[\ta,\tb] &: \ta \tarrow \ta + \tb, \qquad
      &\inr[\ta,\tb] : \tb \tarrow \ta + \tb, 
    \end{aligned} \\ 
    \case[\ta,\tb,\tc] : \ta + \tb \tarrow (\ta \tarrow \tc) \tarrow (\tb \tarrow \tc) \tarrow \tc, \\ 
    \zero : \Nat, \qquad \succ : \Nat \tarrow \Nat, \\
    \totRec[\tc]{n}  : (\Nat \tarrow (\Nat \tarrow \tc) \tarrow \tc) \tarrow \Nat \tarrow \tc. 
  \end{gather*} 
  \begin{legend}
    where \(n\) is a natural number or the symbol \(\infty\). 
    In order we have the constant constructor of type \(\Unit\), the constructor and the two destructors of the product types, the two constructors and the destructor of the sum types and the two constructors and the destructor of the natural type.
    Most of those are actually ``parametric polymorphic'' terms, that is, families of constants indexed by the types \(\ta,\tb\) and \(\tc\).
  \end{legend}
\end{figure}
We use the metavariables \(\tta, \ttb, \ttc\) for terms. 
Moreover for any two terms \(\tta : \ta\) and \(\ttb : \ta \tarrow \tb\) we have a term \(\ttb \tta : \tb\) and for any variable \(\tva : \ta\) and term \(\ttb : \tb\) we have a term \(\qlambda\tva \ttb : \ta \tarrow \tb\). 


In order to avoid a parenthesis overflow, we shall follow the usual conventions for writing terms and types. 
For terms this means that application and abstraction are respectively left and right-associative 
and that abstraction binds as many terms as possible on its right; 
for types it means that \(\times\) and \(+\) are left-associative and associate more closely than \(\tarrow\), which is right-associative. 
We also omit outer parenthesis. 
For example: 
\[ 
  \begin{array}{ccc} 
    \ta \tarrow \tb \tarrow \ta \times \tb \times \tc & \text{ stands for } & (\ta \tarrow (\tb \tarrow ((\ta \times \tb) \times \tc))), \\ 
    \abstr{x}\ta \abstr{y}\tb \abstr{z}\tc t_1 t_2 t_3 & \text{ stands for } & (\abstr\lva\ta (\abstr{y}\tb (\abstr{z}\tc ((t_1 t_2) t_3)))). 
  \end{array} 
\] 

We define the reductions for the terms of system \(T'\):
\begin{gather*} 
  (\abstr\lva\ta t) a \reducesto_\beta t\subst\lva{a}, \\ 
  \begin{aligned} 
    \prl[\ta,\tb] (\pair[\ta,\tb]{a}{b}) &\reducesto_\times a, 
    &\case[\ta,\tb,\tc] (\inl[\ta,\tb] a) f g &\reducesto_+ f a, \\
    \prr[\ta,\tb] (\pair[\ta,\tb]{a}{b}) &\reducesto_\times b, 
    &\case[\ta,\tb,\tc] (\inr[\ta,\tb] b) f g &\reducesto_+ g b, 
  \end{aligned} \\
  \totRec[\tc]{n} h \num m \reducesto_R \begin{cases} 
    h \num m (\totRec[\tc]{m} h) &\text{if } m < n \text{ or } n = \infty, \\ 
    \dummy[\tc] &\text{otherwise,} 
  \end{cases} 
\end{gather*} 
where \( a : \ta\), \(b : \tb\), \(c : \tc\), \(f : \ta \tarrow \tc\), \(g : \tb \tarrow \tc \) and \( h : \Nat \tarrow (\Nat \tarrow \tc ) \tarrow \tc \). 
Note that we use \(c\) as a dummy term of type \(\tc\)\footnotemark. 
\footnotetext{
  As long as the base types are inhabited, we can define an arbitrary dummy term \(\dummy[\ta]\) for any type \(\ta\):
  \begin{gather*}
    \dummy[\Unit] \equiv \unit, \qquad \dummy[\Nat] \equiv \num 0, \\ 
    \dummy[\ta \tarrow \tb] \equiv \abstr\_\ta \dummy[\tb], \quad
    \dummy[\ta \times \tb] \equiv \pair \dummy[\ta] \dummy[\tb], \quad 
    \dummy[\ta + \tb] \equiv \inl \dummy[\ta]. 
  \end{gather*}
}

  We explain the reduction given for \(\totRec{}\), since it is not the standard one. The difference is due to the fact that \(\totRec{}\) is meant to realize complete induction instead of standard induction. 
In complete induction, the inductive hypothesis holds not only for the immediate predecessor of the value we are considering, but also  for all the smaller values. 

Similarly, \(\totRec{}\) allows us to recursively define a function \(f\) where the value of \(f(\num m)\) depends not only on the value of \(f(\num{m-1})\) but also on the value of \(f(\num l)\), for any \(l < m \). 
Thus, when computing \(\totRec[\tc]{n} h \num m\), instead of taking the value of \(\totRec[\tc]{n} h \num{(m-1)}\) as an argument, \(h\) takes the whole function \(\totRec[\tc]{n} h\). 
In order to avoid unbounded recursion, we add a guard \(n\) that prevents \(\totRec[\tc]{n} h\) to be computed on arguments greater or equal to \(n\). 
More precisely \(\totRec[\tc]{n} h \num m \) only reduces to \(h \num m (\totRec[\tc]{m} h) \) if \( m < n \); thus, even if \(h\) requires \(\totRec[\tc]{m} h\) to be computed on many values, the height of the computation trees is bound by \(m\)\footnote{Unlike in standard primitive recursion, where the computation always comprises \(m\) steps, in course-of-values primitive recursion the computation can actually be shorter if \(h\) ``skips'' values.}. 
Naturally, a ``good'' \(h\) should not evaluate \(\totRec[\tc]{m} h\) on values bigger than \(m\), but in any case the guard guarantees termination.
The symbol \(\infty\) acts as a dummy guard, which gets replaced with an effective one when \(\totRec[\tc]{\infty} h\) is evaluated the first time. 
 
 

\begin{wip}
  We can easily generalize our recursion operator to work on arbitrary well-founded types. 
  Assume that \(<_\ta : \ta \tarrow \ta \tarrow \Bool\) is a well-founded ordering on an inhabited type \(\ta\) and choose any term \(0_\ta : \ta\) as a default value. 
  \begin{align*} 
    \totRec[\ta,\tb]\lva : (\ta \tarrow (\ta \tarrow \tb) \tarrow \tb) \tarrow \ta \tarrow \tb \\ 
    \totRec[\ta,\tb]\infty h \lva \reducesto \ite <_\ta 
    \totRec[\ta,\tb]\lva h \lvb \reducesto \ite (<_\ta \lvb \lva) (h \lvb (\totRec[\ta,\tb]\lvb) (0_\ta)
  \end{align*} 
  where \( h : \ta \tarrow \tb \). 
\end{wip}

System \(T'\) shares most of the good properties of G\"odel's system \(T\), in particular confluence, strong normalization\footnotemark and a normal form property. 
\footnotetext{Strong normalization is a consequence of the explicit bound on recursion given by the subscript in the recursion constant.} 

\section{Monadic Realizability} \label[section]{sec:monadic_realizability}
\setmonad\MM

This section contains the abstract part of our work. 
We describe the abstract framework of monadic realizability and show the soundness of \(\HA\) with respect to the semantics induced by a generic monad.


We state the properties that a suitable relation must satisfy in order to be called a monadic realizability relation and 
we show how such a relation induces a (monadic) realizability semantics.
Then we describe the proof decoration procedure to extract monadic realizers from proofs in \(\HA\). 
Here we are only concerned with proofs in \(\HA\), for a non-trivial example of a monadic realizability notion 
see interactive realizability in \cref{sec:monadic_interactive_realizability}. 

We start by introducing a syntactic translation of the concept of monad from category theory. 
Informally, a monad is an operator \(\T\) ``extending'' a type, with a canonical embedding from \(\ta\) to \(\T(\ta)\), a canonical way to lift\fixme{or extend?} a map from \(\ta\) to \(\T(\tb)\) to a map from \(\T(\ta)\) to \(\T(\tb)\), a canonical way of merging an element of \(\T(\ta)\) and an element of \(\T(\tb)\) into an element of \(\T(\ta \times \tb)\). We also requires some equations relating these canonical maps, equations which are often satisfied in the practice of programming. 
\begin{definition}[Syntactic Monad] \label[definition]{def:syntactic_monad} A \emph{syntactic monad} \(\MM\) is a tuple \( (\T\), \(\monUnit\), \(\monStar\), \(\monMerge) \) where \(\T\) is a type constructor, that is, a map from types to types, and, for any types \( \ta, \tb \), 
  \( \monUnit, \monStar \) and \( \monMerge \) are families (indexed by \(\ta\) and \(\tb\)) of closed terms: 
  \begin{align*} 
    \monUnit[\ta] &: \ta \tarrow \T\ta, \\ 
    \monStar[\ta,\tb] &: (\ta \tarrow \T\tb) \tarrow (\T\ta \tarrow \T\tb), \\ 
    \monMerge[\ta,\tb] &: \T\ta \tarrow \T\tb \tarrow \T(\ta \times \tb), 
  \end{align*} 
  satisfying the following properties: 
  \begin{align}
    \tag{M1} \label[property]{mon1} 
    \monStar[\ta,\ta] \monUnit[\ta] \mon\tta &\leadsto \mon\tta, \\ 
    \tag{M2} \label[property]{mon2} 
    \monStar[\ta,\tb] f (\monUnit[\ta] \tta) &\leadsto f \tta, \\ 
    \tag{M3} \label[property]{mon_merge1} 
    \monMerge[\ta,\tb] (\monUnit[\ta] \tta) (\monUnit[\ta] \ttb) &\leadsto \monUnit[\ta\times\tb] (\pair[\ta,\tb]\tta\ttb), 
  \end{align}
  for any \( \mon\tta : \T\ta \), \( f : \ta \tarrow \T\tb \), \( g : \tb \tarrow \T\tc \), \( \tta : \ta \) and \( \tta : \tb \). 
\end{definition} 
The terms \(\monUnit\) and \(\monStar\) and \cref{mon1,mon2} are a straightforward translation of the definition of Kleisli tripe in category theory, an equivalent way to describe a monad\footnotemark. 
\footnotetext{
  This part of the definition follows the one given by Wadler in \cite{wadler92}, with the difference that we replace the term \(\monBind\) with \(\monStar\), where:
  \[ \monBind[\ta,\tb] : \T\ta \tarrow (\ta \tarrow \T\tb) \tarrow \T\tb. \]
  Defining \(\monStar\) and \(\monBind\) in terms of each other is straightforward: 
  \begin{align*}
    \monBind[\ta,\tb] &\equiv 
    \abstr{\mon\tta}{\T\ta} \abstr{f}{\ta \tarrow \T\tb} \monStar f \mon\tta, \\ 
    \monStar[\ta,\tb] &\equiv 
    \abstr{f}{\ta \tarrow \T\tb} \abstr{\mon\tta}{\T\ta} \monBind \mon\tta f. 
  \end{align*}
  The term \(\monStar\) corresponds directly to the operator \(\_^*\) in the definition of Kleisli triple. 
}

Term \(\monMerge\) and \cref{mon_merge1} are connected to the definition of strong monad: \(\monMerge\) is the syntactical counterpart of the natural transformation \(\phi\), induced by the tensorial strength of the monad (see \cite{moggi91} for details). 
While \(\phi\) satisfies several other properties in \cite{moggi91}, \cref{mon_merge1} is the only one we need for our treatment. \note{\(\phi\) may be related to commutative strong monad}

\begin{example} \label{ex:identity_monad1}
  \setmonad\IdM 
  The simplest example of syntactic monad is the \emph{identity monad} \(\IdM\), defined as:  
  \begin{align*}
    \T\ta &\equiv \ta, & 
    \monUnit[\ta] &\equiv \abstr\tta\ta \tta, \\
    \monStar[\ta,\tb] &\equiv \abstr{f}{\ta \tarrow \tb} f, & 
    \monMerge[\ta,\tb] &\equiv \pair[\ta,\tb]. 
  \end{align*}
  This monad cannot describe any additional computational property besides the value a term reduces to. 
\end{example}

A \emph{realizability relation} is a binary relation between terms and closed formulas. 
When a term and a formula are in such a relation we shall say that the term \emph{realizes} the formula or that the term is a \emph{realizer} of the formula. 
The intended meaning is that a realizer of a formula is the computational content of a proof of the formula. 

We proceed towards the definition of a family of realizability relations, which we call \emph{monadic realizability relations}.
Any monadic realizability relation is given with respect to some monad \(\MM\) and determines a particular notion of realizability where realizers have the computational properties described by the monad. 
In the rest of this section we shall assume that \(\MM = (\T, \monUnit, \monStar, \monMerge)\) denotes any fixed syntactic monad. 

We now define the type of the monadic realizers of a formula. 
The idea is to take the standard definition of the type of intuitionistic realizers of a formula \(\fa\) and to apply \(\T\) only to the type \(\ta\) of the whole formula \(\fa\) and to the types appearing in \(\ta\) after an arrow, 
namely the types of consequents \(\fc\) of implication sub-formulas \(\fb \limply \fc\) in \(\fa\) and the types of bodies \(\fb\) of universal quantified sub-formulas \(\qforall\lva \fb\) in \(\fa\). 
This is the standard call-by-value way to treat arrow types in a monadic framework explained in \cite{wadler90}. 

\begin{definition}[Types for Monadic Realizers] \label[definition]{def:mon_types}
  We define two mappings \(\mm\cdot\) and \(\m\cdot\) from formulas to types by simultaneous recursion. 
  The first is the outer or monadic typing of a formula \(\fa\): 
  \[ \mm\fa = \T\m\fa, \] 
  and the latter is the inner typing, defined by induction on the structure of \(\fa\): 
  \begin{align*} 
    \m\lafa & = \Unit, &
    \m{\fb \land \fc} & = \m\fb \times \m\fc, \\ 
    \m{\fb \lor \fc} & = \m\fb + \m\fc, &
    \m{\qexists\lva \fb} & = \Nat \times \m\fb, \\ 
    \m{\fb \limply \fc} & = \m\fb \tarrow \mm\fc, &
    \m{\qforall\lva \fb} & = \Nat \tarrow \mm\fb, 
  \end{align*} 
  where \(\lafa\) is an atomic formula and \(\fa\) and \(\fb\) are any formulas. 
\end{definition} 
We consider \(\lfalse\) to be atomic and \(\lnot \fa\) to be a notation for \(\fa \limply \lfalse\), so the types of their realizers follow from the previous definition. 

As we defined two types for each formula \(\fa\), each formula has two possible realizers, one of type \(\m\fa\) and one of type \(\mm\fa\). 
The former will follow the BHK interpretation like an ordinary intuitionistic realizer 
while the latter will be able to take advantage of the computational properties given by the syntactic monad \(\MM\). 
A formula (in particular classical principles) may have a realizer of monadic type but no realizer of inner type. 

We shall now state the requirements for a realizability relation to be a monadic realizability relation. 
A realizability relation is to be thought of as the restriction of the realizability semantics to closed formulas, 
that is, a relation between terms of \(T'\) and closed formulas which holds when a term is a realizer of the formula. 
Since a formula can have realizers of inner and outer type, in the following definition two realizability relations will appear: \(\re\) for realizers of inner type, whose definition is modeled after the BHK interpretation and \(\monRe\) for the realizers of outer type, which takes in consideration the computational properties of the monad \(\MM\). 

As a typographical convention we shall use the letters \(\rr\), \(\ra\) and \(\rb\) for terms of type \(\m\fa\). 
Similarly we shall use \(\mon\rr\), \(\mon\ra\) and \(\mon\rb\) for terms of type \(\mm\fa\). 
\begin{definition}[Monadic Realizability Relation] \label[definition]{def:monadic_realizability_relation}
  Let \(\monRe\) be a realizability relation between terms of type \(\mm\fa\) and closed formulas \(\fa\). 
  Let \(\re\) be another realizability relation between terms of type \(\m\fa\) and closed formulas \(\fa\), such that
  \begin{itemize}
    \item \label[property]{re_atomic} \( \rr \re \lafa \) iff 
      \( \rr \leadsto \unit\) and \(\lafa\) is true, 
    \item \label[property]{re_and} \( \rr \re \fb \land \fc \) iff 
      \(\prl \rr \re \fb \) and \(\prr r \re \fc \), 
    \item \label[property]{re_or} \( \rr \re \fb \lor \fc \) iff 
      \(\rr \leadsto \inl a \) and \( a \re \fb \) or 
      \(r \leadsto \inr b \) and \( b \re \fc \), 
    \item \label[property]{re_imply} \( \rr \re \fb \limply \fc \) iff 
      \( \rr \ra \monRe \fc \) for all \(\ra : \m\fb\) such that \( \ra \re \fb \), 
    \item \label[property]{re_forall} \( \rr \re \qforall{x} \fb \) iff 
      \(\rr \num n \monRe \fb\subst{x}{\num n} \) for all natural numbers \(n\), 
    \item \label[property]{re_exists} \( \rr \re \qexists{x} \fb \) iff 
      \(\prr \rr \re \fb\subst{x}{\prl \rr} \), 
  \end{itemize} 
  where \(\lafa\) is a closed atomic formula and \(\fb\) and \(\fc\) are generic formulas. 
  We consider \(\lfalse\) a closed atomic formula which is never true (for instance \(0=1\)).
  We shall say that the pair \( (\monRe, \re) \) is a \emph{monadic realizability relation} if the following properties are satisfied: 
  \begin{enumerate}[label=MR\arabic*]
    \item \label[property]{real1} 
      if \( \rr \re \fa \) then \( \monUnit \rr \monRe \fa \), 
    \item \label[property]{real2} 
      if \( \rr \re \fb \limply \fc \)
      then \( \monStar \rr \mon\ra \monRe \fc \) for all \( \mon\ra : \mm\fb \) such that \( \mon\ra \monRe \fb \), 
    \item \label[property]{real3} 
      if \( \mon\ra \monRe \fb \) and \( \mon\rb \monRe \fc \) 
      then \( \monMerge \mon\ra \mon\rb \monRe \fb \land \fc \). 
  \end{enumerate} 
  We will say that a term \(\rr\) (resp. \(\mon\rr\)) is an \emph{inner} (resp. \emph{outer} or \emph{monadic}) \emph{realizer} of a formula \(\fa\) if \(\rr : \m\fa\) (resp. \(\rr : \mm\fa\)) and \(\rr \re \fa\) (resp. \(\mon\rr \monRe \fa\)). 
\end{definition}
When defining a concrete monadic realizability relation, it is often convenient to define \(\monRe\) in terms of \(\re\) too, that is, the two relations will be defined by simultaneous recursion in terms of each other. 

Note how the properties of the relation \(\re\) resemble the clauses the definition of standard modified realizability. 
The main difference is that in the functional cases, those of implication and universal quantification, \(\re\) is not defined in terms of itself but uses \(\monRe\). 
This makes apparent our claim that the behavior of inner realizers is closely related to the BHK interpretation. 

\Cref{real1} is a constraint on the relationship between \(\monRe\) and \(\re\). 
It requires \(\monUnit\) to transform inner realizers into monadic realizers, 
which can be thought as the fact that realizers satisfying the BHK interpretation are acceptable monadic realizers. 
\Cref{real2} again links \(\re\) and \(\monRe\), this time through \(\monStar\). 
It says that, if we have a term that maps inner realizers into monadic realizers, its lifting by means of \(\monStar\) maps monadic realizers into monadic realizers. 
\Cref{real3} is a compatibility condition between \(\monMerge\) and \(\monRe\). 
These conditions are all we shall need in order to show that any monadic realizability relation determines a sound semantics for \(\HA\). 
Later we shall see how particular instances of monadic realizability 
can produce a sound semantics for more than just \(\HA\). 

\begin{example}
  \setmonad\IdM
  We continue our example with the identity monad \(\IdM\) by defining a monadic realizability relation. 
  We define \(\monRe\) and \(\re\) by simultaneous recursion, with \(\re\) defined in terms of \(\monRe\) as in \cref{def:monadic_realizability_relation} and \(\monRe\) defined as \(\re\), which makes sense since \( \mm\fa = \m\fa\). 
\end{example}

We can now define the monadic realizability semantics for a given monadic realizability relation, that is, we say when a realizer validates a sequent where a formula can be open and depend on assumptions. 
In order to do this we need a notation for a formula in a context, which we call \emph{decorated sequent}. 
A decorated sequent has the form \( \Gamma \monSeq \rr : \fa \) where \(\fa\) is a formula, \(\rr\) is a term of type \(\mm\fa\) and 
\(\Gamma\) is the context, namely, a list of assumptions written as \( \alpha_1 : \fa_1, \dotsc \alpha_k : \fa_k \) where \(\fa_1, \dotsc, \fa_k\) are formulas and \(\alpha_1, \dotsc, \alpha_k\) are proof variables that label each assumption, that is, they are variables of type \( \m{\fa_1}, \dotsc, \m{\fa_k} \). 
As we did with the syntactic monad \(\MM\), in the following we shall assume to be working with a fixed generic monadic realizability relation \(\monRe\). 
\begin{definition}[Monadic Realizability Semantics] \label[definition]{def:mon_sem}
  Consider a decorated sequent: 
  \[ \alpha_1 : \fa_1, \dotsc, \alpha_k : \fa_k \monSeq \mon\rr : \fb, \] 
  such that the free variables of \(\fb\) are \( x_1, \dotsc, x_l \) and the free variables of \(\mon\rr\) are either in \( x_1, \dotsc, x_l \) or in \( \alpha_1, \dotsc, \alpha_k \). 
  We say that the sequent is valid if and only if 
  for all natural numbers \(n_1, \dotsc, n_l\) and for all inner realizers \( \ra_1 : \m{\fa_1}, \dotsc, \ra_k : \m{\fa_k} \) such that 
  \[ \ra_1 \re \fa_1[x_1 \coloneqq \num{n}_1, \dotsc, x_l \coloneqq \num{n}_l] \qquad \dotso \qquad \ra_k \re \fa_k[x_1 \coloneqq \num{n}_1, \dotsc, x_l \coloneqq \num{n}_l], \] 
  we have that 
  \[ \mon\rr[x_1 \coloneqq \num{n}_1, \dotsc, x_l \coloneqq \num{n}_l, \alpha_1 \coloneqq \ra_1, \dotsc, \alpha_k \coloneqq \ra_k] \monRe A[x_1 \coloneqq \num{n}_1, \dotsc, x_l \coloneqq \num{n}_l]. \] 
\end{definition} 

\begin{example}\setmonad\IdM
  From \cref{def:mon_sem}, it follows that the semantics induced by the monadic realizability relation \(\monRe\) is exactly the standard semantics of modified realizability.  
\end{example}

Now that we have defined our semantics, we can illustrate the method to extract monadic realizers from proofs in \(\HA\). Later we shall show how to extend our proof extraction technique to \(\HA+\EM\). 
Since proof in \(\HA\) are constructive, the monadic realizers obtained from them behave much like their counterparts in standard modified realizability and comply with the BHK interpretation. 
In \cref{sec:monadic_interactive_realizability} we shall show how to extend the proof decoration to non constructive proofs by exhibiting a monadic realizer of \(\EM\) that truly takes advantage of monadic realizability since it does not act accordingly to the BHK interpretation. 

In order to build monadic realizers of proofs in \(\HA\) 
we need a generalization of \(\monStar\) that works for functions of more than one argument. 
We can build it using \(\monMerge\) to pack realizers together. 
Thus let
\begin{equation*} 
  \monStarN[\ta_1,\dotsc,\ta_k,\tb]{k} : (\ta_1 \tarrow \dotsb \tarrow \ta_k \tarrow \T\tb) \tarrow (\T\ta_1 \tarrow \dotsb \tarrow \T\ta_k \tarrow \T\tb), 
\end{equation*}
be a family of terms defined by induction on \( k \ge 0 \): 
\begin{gather*}
  \monStarN[\tb]{0}  \equiv \abstr{f}{\T\tb} f, \qquad 
  \monStarN[\ta,\tb]{1}  \equiv \monStar[\ta,\tb], \\ 
  \monStarN{k+2} \equiv \abstr{f}{\ta_1 \tarrow \dotsb \tarrow \ta_{k+2} \tarrow \T\tb} \abstr{x}{\T\ta_1} \abstr{y}{\T\ta_2} \monStarN{k+1} (\abstr{z}{\ta_1 \times \ta_2} f (\prl z) (\prr z) ) (\monMerge x y). 
\end{gather*} 
For instance: 
\[ \monStarN{2} \equiv \abstr{f}{\ta \tarrow \tb \tarrow \T\tc} \abstr{x}{\T\ta} \abstr{y}{\T\tb} \monStar (\abstr{z}{\ta \times \tb} f (\prl z) (\prr z)) (\monMerge x y) \] 
\begin{omitted}
  An alternative definition of \(\monStarN{k}\) that does not use \(\monMerge\): 
  \[ \monStarN{k} \equiv \abstr{f}{\ta_1 \tarrow \dotsb \tarrow \ta_k \tarrow \T\tb}
    \abstr{\mon\tta_1}{\T\ta_1} \dotsc \abstr{\mon\tta_k}{\T\ta_k} 
  \monStar (\abstr{\tta_1}\ta \monStarN{k-1} (f \tta_1) \mon\tta_2 \dotsm \mon\tta_k ) \mon \tta_1\] 
\end{omitted}

Moreover we shall need to ``raise'' the return value of a term \( f : \ta_1 \tarrow \dotsb \tarrow \ta_k \tarrow \tb \) with \(\monUnit\) before we apply \(\monStarN{k}\).
We define the family of terms \(\monRaiseN{k}\) by means of \(\monStarN{k}\), for any \(k \ge 0\):
\begin{align*}
  \monRaiseN{k} &: (\ta_1 \tarrow \dotsb \tarrow \ta_k \tarrow \tb) \tarrow (\T\ta_1 \tarrow \dotsb \tarrow \T\ta_k \tarrow \T\tb) \\ 
  \monRaiseN{k} &\equiv \abstr{f}{\ta_1 \tarrow \dotsb \tarrow \ta_k \tarrow \tc} \monStarN{k} 
  (\abstr{x_1}{\ta_1} \dotsm \abstr{x_k}{\ta_k} \monUnit (f x_1 \dotsm x_k)), 
\end{align*}

Now we can show how to extract a monadic realizer from a proof in \(\HA\). 
Let \(\mathcal{D}\) be a derivation of some formula \(\fa\) in \(\HA\), that is, a derivation ending with \( \Gamma \seq \fa \). 
We produce a decorated derivation by replacing each rule instance in \(\mathcal{D}\) with the suitable instance of the decorated version of the same rule given in \cref{fig:monadic_decorated_rules}. 
These decorated rules differ from the previous version in that they replace sequents with decorated sequents, that is, they bind a term to each formula, 
where the term bound to the conclusion of a rule is build from the terms bound to the premises. 
Thus we have defined a term by structural induction on the derivation: if the conclusion of the decorated derivation is \( \Gamma \monSeq \mon\rr : \fa \) then we set \( \mathcal{D}^* \equiv \mon\rr \). 

\begin{figure}[!ht]
  \begin{mdframed}
  \caption{\(\HA\) rules, decorated with monadic realizers.} \label{fig:monadic_decorated_rules}
  \begin{gather*}
      \renewcommand\arraystretch{2}
      \newcommand\wline[1]{\multicolumn{2}{c}{#1} \\ }
      \PrAx{} 
      \PrLbl{\text{Id}} 
      \PrUn{\Gamma \monSeq \monRaiseN{0} \lva : \fa} 
      \DisplayProof \qquad 
      \PrAx{\Gamma \monSeq \mon\rr_1 : \lafa_1} 
      \PrAx\dotso 
      \PrAx{\Gamma \monSeq \mon\rr_l : \lafa_l} 
      \PrLbl{\text{Atm}}
      \PrTri{\Gamma \monSeq \monRaiseN{l} (\abstr{\gamma_1}\Unit \dotsm \abstr{\gamma_l}\Unit \unit) \mon\rr_1 \dotsm \mon\rr_l : \lafa} 
      \DisplayProof \\ 
        \PrAx{\Gamma \monSeq \mon\rr_1 : \fa} 
        \PrAx{\Gamma \monSeq \mon\rr_2 : \fb} 
        \PrLbl{\RuleNameI\land} 
        \PrBin{\Gamma \monSeq \monRaiseN{2} \pair \mon\rr_1 \mon\rr_2 : \fa \land \fb} 
        \DisplayProof \\
        \PrAx{\Gamma \monSeq \mon\rr : \fa \land \fb} 
        \PrLbl{\RuleName[L]\land{E}}
        \PrUn{\Gamma \monSeq \monRaiseN{1} \prl \mon\rr : \fa} 
        \DisplayProof \qquad 
        \PrAx{\Gamma \monSeq \mon\rr : \fa \land \fb} 
        \PrLbl{\RuleName[R]\land{E}}
        \PrUn{\Gamma \monSeq \monRaiseN{1} \prr \mon\rr : \fb} 
        \DisplayProof \\ 
      \PrAx{\Gamma \monSeq \mon\rr_1 : \fa} 
      \PrLbl{\RuleName[R]\lor{I}} 
      \PrUn{\Gamma \monSeq \monRaiseN{1} \inl \mon\rr_1 : \fa \lor \fb} 
      \DisplayProof \qquad 
      \PrAx{\Gamma \monSeq \mon\rr_2 : \fb} 
      \PrLbl{\RuleName[L]\lor{I}} 
      \PrUn{\Gamma \monSeq \monRaiseN{1} \inr \mon\rr_2 : \fa \lor \fb} 
      \DisplayProof \\ 
        \PrAx{\Gamma \monSeq \mon\rr : \fa \lor \fb } 
        \PrAx{\Gamma, \alpha_{k+1} : \fa \monSeq \mon\ra : \fc} 
        \PrAx{\Gamma, \alpha_{k+1} : \fb \monSeq \mon\rb : \fc} 
        \PrLbl{\RuleNameE\lor} 
        \PrTri{\Gamma \monSeq \monStarN{1} (\abstr\gamma{\m\fa + \m\fb} \case \gamma (\abstr{\alpha_{k+1}}{\m\fa} \mon\ra) (\abstr{\alpha_{k+1}}{\m\fb} \mon\rb)) \mon\rr : \fc} 
        \DisplayProof \\ 
      \PrAx{\Gamma, \alpha_{k+1} : \fa \monSeq \mon\rr : \fb} 
      \PrLbl{\RuleNameI\limply}
      \PrUn{\Gamma \monSeq \monRaiseN{0} (\abstr{\alpha_{k+1}}{\m\fa} \mon\rr) : \fa \limply \fb} 
      \DisplayProof \quad 
      \PrAx{\Gamma \monSeq \mon\rr : \fa \limply \fb} 
      \PrAx{\Gamma \monSeq \mon\ra : \fa} 
      \PrLbl{\RuleNameE\limply}
      \PrBin{\Gamma \monSeq \monStarN{2} (\abstr{\gamma_1}{\m\fa \tarrow \m\fb} \abstr{\gamma_2}{\m\fa} \gamma_1 \gamma_2) \mon\rr \mon\ra : \fb} 
      \DisplayProof \\ 
      \PrAx{\Gamma \monSeq \mon\rr : \fa} 
      \PrLbl{\RuleNameI\forall}
      \PrUn{\Gamma \monSeq \monRaiseN{0} (\abstr\lva\Nat \mon\rr) : \qforall\lva \fa} 
      \DisplayProof \qquad 
      \PrAx{\Gamma \monSeq \mon\rr : \qforall\lva \fa} 
      \PrLbl{\RuleNameE\forall}
      \PrUn{\Gamma \monSeq (\monStarN{1} (\abstr\gamma{\Nat \tarrow \mm\fa} \gamma \lta)) \mon\rr : \fa\subst\lva\lta} 
      \DisplayProof \\ 
        \PrAx{\Gamma \monSeq \mon\rr : \fa\subst\lva\lta} 
        \PrLbl{\RuleNameI\exists}
        \PrUn{\Gamma \monSeq \monRaiseN{1} (\abstr\gamma{\m\fa} \pair\lta\gamma) \mon\rr : \qexists\lva \fa} 
        \DisplayProof \\ 
        \PrAx{\Gamma \monSeq \rr_1 : \qexists\lva \fa} 
        \PrAx{\Gamma, \alpha : \fa\subst\lva\lvb \monSeq \rr_2 : \fc} 
        \PrLbl{\RuleNameE\exists}
        \PrBin{\Gamma \monSeq \monStarN{1} (\abstr\gamma{\Nat \times \m\fa} (\abstr{y}\Nat \abstr\alpha{\m\fa} r_2)(\prl \gamma)(\prr \gamma)) \rr_1 : \fc} 
        \DisplayProof \\ 
        \PrAx{\Gamma, \alpha_{k+1} : \qforall\lvc \lvc<\lvb \limply \fa\subst\lva\lvc \monSeq \rr : \fa\subst\lva\lvb} 
        \PrLbl{\text{Ind}}
        \PrUn{\Gamma \monSeq \monRaiseN{0} (\totRec\infty f) : \qforall\lva \fa} 
        \DisplayProof 
    \end{gather*}
    \begin{legend}
    where all formulas in rule Atm are atomic, \(\lta\) is any term and
    \(f\) is defined as follows: 
    \[ f \equiv \abstr{y}\Nat \abstr\beta{\Nat \tarrow \T\m\fa} 
      (\abstr\alpha{\Nat \tarrow \T(\Unit \tarrow \T\m\fa)} r) 
      (\abstr{z}\Nat \monRaiseN{0} 
    (\abstr{\_}\Unit \beta z)), \] 
    with \(\beta\) not free in \(r\). 
    \end{legend}
  \end{mdframed}
\end{figure}

In \cref{fig:monadic_decorated_rules},
the rule labeled Atm shows how to decorate any atomic rule of \(\HA\). 
By definition unfolding, we may check that an atomic rule is interpreted as a kind of ``merging'' of the information associated to each premise. 
The nature of the merging depends on the monad we choose.

Note how the monadic realizer of each rule is obtained by lifting the suitable term in the corresponding standard modified realizer with \(\monStarN k\) or \(\monRaiseN k\). 
These monadic realizers do not take advantages of particular monadic features (it cannot be otherwise since we have made no assumption on the syntactic monad or the monadic realizability relation). 
The main difference is that they can act as ``glue'' between ``true'' monadic realizers of non constructive axioms and rules, for instance the one we shall build in \cref{sec:monadic_interactive_realizability}. 

Here we can see that monadic realizability generalizes intuitionistic realizability: decorated rules in \cref{fig:monadic_decorated_rules} reduce to the standard decorated rules for intuitionistic modified realizability in the case of the identity monad \(\IdM\). 

Now we can prove that \(\HA\) is sound with respect to the monadic realizability semantics given in \cref{def:mon_sem}. 
This amounts to say that we can use proof decoration to extract, from any proof in \(\HA\), a monadic realizer that makes its conclusion valid. 
We prove this for a generic monad, which means that the soundness of \(\HA\) does not depend on the special properties of any specific monad. 
The proof only needs the simple properties we have requested in \cref{def:monadic_realizability_relation}. 
\begin{theorem}[Soundness of \(\HA\) with respect to the Monadic Realizability Semantics] \label{thm:ha_soundness}
  Let \(\mathcal{D}\) be a derivation of \( \Gamma \seq \fa \) in \(\HA\) and \(\monRe\) a monadic realizability relation. 
  Then \( \Gamma \monSeq \mathcal{D}^* : \fa \) is valid with respect to \(\monRe\). 
\end{theorem}
The proof is long but simple, proceeding by induction on the structure of the decorated version of \(\mathcal{D}\). 

\Cref{thm:ha_soundness} entails that any specific monadic realizability notion is a sound semantics for at least \(\HA\). 
Later, when we prove that \(\HA+\EM\) is sound with respect to interactive realizability semantics, 
we will only need to show that \(\EM\) is sound since the soundness of \(\HA\) derives from \cref{thm:ha_soundness}.

\section{Monadic Interactive Realizability} \label{sec:monadic_interactive_realizability} 

\setmonad\IR

In this section we define interactive realizability as a particular notion of monadic realizability. 
Thus we show that monadic realizability may realize a sub-classical principle, in this case excluded middle restricted to semi-decidable statements. 

In order to describe the computational properties of interactive realizability (see \cite{aschieriB10}) we need to define a suitable monad. 
As we said, interactive realizability is based on the idea of learning by trial and error. 
We express the idea of trial and error with an exception monad: a term of intended type \(\ta\) has actual type \(\ta + \Ex\), where \(\Ex\) is the type of exceptions, so that a computation may either return its intended value or an exception. 
The learning part, which is described by the dependency on a knowledge state, fits with a part of the side-effects monad (see \cite{moggi91} for more details): a term of intended type \(\ta\) has actual type \(\State \tarrow \ta\), where \(\State\) is the type of knowledge states, so that the value of a computation may change with the state. 
The syntactic monad we are about to define for interactive realizability combines these two monads. 

We introduce \(\Ex\) and \(\State\) as base types and some term constants satisfying suitable properties. 
Actually, system \(T'\) is expressive enough to explicitly define \(\Ex\) and \(\State\) and the terms we need, but we prefer a cleaner abstract approach. 
Therefore, we explain the intended meaning of \(\Ex\) and \(\State\) and use it in the following as a guideline. 

\newcommand\LRel{\mathcal{R}}
We write \(\LRel_k\) for the set of symbols of the \(k\)-ary predicates in \(\HA\). 
The intended interpretation of a (knowledge) state \(s\) is a partial function
\[ \interp s : \left( \bigcup_{k=0}^\infty \LRel_{k+1} \times \N^k \right) \rightharpoonup \N, \] 
that sends a \(k+1\)-ary predicate symbol \(P\) and a \(k\)-tuple of parameters \( m_1, \dotsc, m_k \in \N \) to a witness for \(\qexists{x} P (\num m_1, \dotsc, \num m_k, x) \).
We interpret the fact that a state \(s\) is undefined for some \(P,m_1,\dotsc,m_k\) as a lack of knowledge about a suitable witness. 
This is either due to the state being incomplete, meaning that there exists a suitable witness \(m\) we could use to extend the state by setting \( s(P,(m_1,\dotsc,m_k)) = m\), or to the fact that there are no suitable witness, meaning that \( \qforall{x} \lnot P(\num m_1, \dotsc, \num m_k, x) \) holds\footnote{Here we are using \(\EM\) at the metalevel in order to explain the possible situations. Using a principle at the metalevel in order to justify the same principle in the logic is a common practice. In our treatment this is not problematic because we never claim to be able to \emph{effectively decide} which situation we are in.}. 
We require that \(s\) satisfies two properties. 
The first is for \(s\) to be \emph{sound}, meaning that its values are actually witnesses. More precisely: 
\[ \interp s (P, (m_1, \dotsc, m_k)) = m \text{ entails } P(\num m_1, \dotsc, \num m_k, \num m). \]
The second is that \(s\) is \emph{finite}, namely that the domain of \(s\) (the set of values \(s\) is defined on) is finite. 
This because we want a knowledge state to encode a finite quantity of information. 
Let \(\interp\State\), the set of all finite sound states, be the intended interpretation of the type \(\State\). 
Recall that there is a canonical partial order on states given by the extension relation: we write \(s_1 \leq s_2\) and read ``\(s_2\) \emph{extends} \(s_1\)'' if and only if \(s_2\) is defined whenever \(s_1\) is and with the same value. 

An exception \(e : \Ex \) is produced when we instantiate an assumption of the form \( \qforall{x} \allowbreak \lnot P(\num m_1, \dotsc, \allowbreak \num m_k, x) \) with some \(m\) such that \( \lnot P(\num m_1, \dotsc, \num m_k, \num m) \) does not actually hold (remember that we proceed by trial and error, in particular we may assume things that are actually false). 
This means that \(m\) is a witness for \(\qexists{x} P (\num m_1, \dotsc, \num m_k, x) \), in particular it could be used to extend the knowledge state on values where it was previously undefined. 
The role of exceptions is to encode information about the discovery of new witnesses: since we use this information to extend states the intended interpretation of an exception \(e\) is as a partial function: 
\[ \interp e : \interp\State \rightharpoonup \interp\State. \]
Since \(e\) extends states we require that \( s \leq e(s) \). 
We interpret an exception as a partial function because an exception \(e\) may fail to extend some state \(s\). 
The reason is that \(e\) may contain information about a witness \(m'\) for an existential statement \(\qexists{x} P (\num m_1, \dotsc, \num m_k, x) \) on which \(s\) is already defined as \(m\). 
Note that an existential formula can have more that one witness so two cases may arise: 
either \(m = m'\), meaning that the information of \(e\) is already part of \(s\) or \(m \neq m'\) so that the information of \(e\) is incompatible with the information of the state. 
In the first case \( e(s) = s \), while in the second case \( e(s) \) is not defined. 

Before defining the syntactic monad \(\IR\) for interactive realizability, we need to introduce some terminology on exceptions and states. 
\begin{definition}[Terminology on Exceptions and States]
  We say that a term of type \(\ta + \Ex\) is either a \emph{regular} value \(a\) if it reduces to \(\inl a\) for some term \(a : \ta\) or an \emph{exceptional} value if it reduces to \(\inr e\) for some term \(e : \Ex\). 
  We say that a term of type \(\State \tarrow \ta\) is a \emph{state function}. 
  Finally we say that an exception \(e\) \emph{properly extends} \(s\) if \(e(s)\) is defined and \( s < e(s) \). 
\end{definition}

Note that different exceptions might be used to extend a knowledge state in incompatible ways, that is, by sending the same predicate symbol and the same tuple of parameters into different witnesses. 
In order to mediate these conflicts, we introduce the term constant: 
\[ \exmerge : \Ex \tarrow \Ex \tarrow \Ex. \]
The role of the \(\exmerge\) function is to put together the information from two exceptions into a single exception. 
This means that \(\exmerge\) cannot simply put together all the information from its argument: if such information contains more that one distinct witness for the same existential statement it must choose one in some arbitrary way, for instance the leftmost or the minimum witness.
Many choices for \(\exmerge\) are possible, provided that they satisfy the following property:
\begin{equation}\label[property]{exmerge_prop} \tag{EX}
  \left. \begin{array}{r} 
    e_1 \text{ properly extends } s \\ 
e_2 \text{ properly extends } s \end{array} \right\} 
\text{ entails that } \exmerge e_1 e_2 \text{ properly extends } s,
\end{equation}
for any state \(s\) and exceptions \(e_1, e_2\).

Before the definition we give an informal description of \(\IR\). 
The monad \(\IR\) maps a type \(\ta\) to \(\State \tarrow (\ta + \Ex)\), that is, values of type \(\ta\) are lifted to state functions that can throw exceptions. 
The term \(\monUnit\) maps a value \(a : \ta\) to a constant state function that returns the regular value \(a\). 
If \(f : \ta \tarrow \T\tb\) then \(\monStar f\) is a function with two arguments, a state \(s\) and a state function \(\mon a : \T\ta\). 
It evaluates \(\mon a\) on \(s\): if this results in a regular value \(a : \ta\) it applies \(f\) to \(a\), otherwise it propagates the exceptional value. 
Finally, if \(\mon a : \T\ta\) and \(\mon b : \T\tb\) are two state functions, then \(\monMerge \mon a \mon b\) is a state function that evaluates its arguments on its state argument:
when both arguments are regular values it returns their pair, otherwise it propagates the exception(s), using \(\exmerge\) if both arguments are exceptional values. 

We are now ready to give the formal definition of \(\IR\). 
\begin{definition}[Interactive Realizability Monad] \label{def:ir_monad}
  Let \(\IR\) be the tuple \( (\T\), \(\monUnit\), \(\monStar\), \(\monMerge)\), where 
  \begin{align*} 
    \T \ta 
    &= \State \tarrow (\ta + \Ex), \\ 
    \monUnit[\ta] 
    &\equiv \abstr\tta\ta \abstrS{\_} \inl[\ta,\Ex] \tta, \\ 
    \monStar[\ta,\tb] 
    &\equiv \abstr{f}{\ta \tarrow \T\tb} \abstr{\mon\tta}{\T\ta} 
    \abstrS{s} \case[\ta,\Ex,\tb+\Ex] (\mon\tta s) 
    (\abstr\tta\ta f \tta s) \inr[\tb,\Ex], \\  
    \monMerge[\ta,\tb] 
    &\equiv \abstr{\mon\tta}{\T\ta} \abstr{\mon\ttb}{\T\tb} \abstrS{s} \case[\ta,\Ex,(\ta \times \tb)+\Ex] (\mon \tta s) \\ 
    &\phantomrel\equiv (\abstr\tta\ta \case[\tb,\Ex,(\ta \times \tb)+\Ex] (\mon\ttb s) (\abstr\ttb\tb \inl[\ta \times \tb,\Ex] (\pair \tta \ttb)) \inr[\ta \times \tb,\Ex]) \\ 
    &\phantomrel\equiv (\abstr{e_1}\Ex \case[\tb,\Ex,(\ta \times \tb)+\Ex] (\mon\ttb s) (\abstr\_\tb \inr[\ta \times \tb,\Ex] e_1) (\abstr{e_2}\Ex \inr[\ta \times \tb,\Ex] (\exmerge e_1 e_2))),
  \end{align*} 
  for some \(\exmerge\) satisfying \cref{exmerge_prop}. 
\end{definition} 

The term \(\monUnit[\ta]\) takes a value \(a : \ta\) and produces a constant state function that returns the regular (as opposed to exceptional) value \(a\). 
The term \(\monStar[\ta,\tb]\) takes a function \(f : \ta \tarrow \T\tb\) and returns a function \(f'\) which lifts the domain of \(f\) to \(\T\ta\). 
The state function returned by \(f' \) when applied to some \(\mon a : \T\ta\) behaves as follows: it evaluates \(\mon a\) on the state and if \(\mon a s\) is a regular value \(a : \ta\) it returns \(f a\); otherwise if \(\mon a s\) is an exception it simply propagates the exception. 
The term \(\monMerge[\ta,\tb]\) takes two state functions \(\mon a : \T\ta\) and \(\mon b : \T\tb\) and returns a state function \( \mon c : \T(\ta \times \tb)\). 
When both arguments are regular values it returns their pair, otherwise it propagates the exception(s), using \(\exmerge\) if both arguments are exceptional. 

We omit the proof of the fact that \(\IR\) is a syntactic monad since it is a simple verification. 


We now define a family of monadic realizability relations, one for each state \(s\), requiring that a realizer, applied to a knowledge state \(s\), 
either realizes a formula in the sense of the BHK semantics or can extend \(s\) with new knowledge.
\begin{definition}[Interactive Realizability Relation] 
  Let \(s\) be a state, \( \mon\rr : \mm\fa \) be a term and \(\fa\) a closed formula. 
  We define two realizability relations \(\monRe^s\) and \(\re^s\) by simultaneous induction on the structure of \(\fa\): 
  \begin{itemize}
    \item \(\mon\rr \monRe^s \fa \) if and only if we have that \(\mon\rr s\) is either a regular value \(\rr\) such that \(\rr \re^s \fa\) or an exceptional value \(e\) such that \(e\) properly extends \(s\), 
    \item \(\re^s\) is defined in terms of \(\monRe^s\) by the clauses in \cref{def:monadic_realizability_relation}. 
  \end{itemize}
  We say that \(\mon\rr\) (resp. \(\rr\)) is a \emph{monadic} (resp. \emph{inner}) \emph{interactive realizer} of \(\fa\) with respect to \(s\) when \(\mon\rr : \mm\fa\) (resp. \(\rr : \m\fa\)) and \(\mon\rr \monRe^s \fa\) (resp. \(\rr \re^s \fa\)).
\end{definition}

In order to show that any interactive realizability relations with respect to a state is a monadic realizability relation we need to verify that is satisfies the required properties. 
\begin{lemma}[The Monadic Realizability Relation \(\monRe^s\)] \label{thm:ir_monadic_realizability}
  For any state \(s\), \(\monRe^s\) is a monadic realizability relation. 
\end{lemma} 
Following \cref{def:mon_sem}, for each state \(s\), the monadic realizability relation \(\monRe^s\) induces a monadic realization semantics, which realizes \(\HA\) by \cref{thm:ha_soundness}.  
We employ this family of semantics indexed by a state in order to define another one, which does not depend on a state. 
\begin{definition}[Interactive Realizability Semantics] \label{def:interactive_realizability_semantics}
  We say that the decorated sequent \(\Gamma \monSeq \mon\rr : \fa\) is valid if and only if it is valid with respect to the semantics induced by each \(\monRe^s\) for every state \(s\). 
\end{definition} 
We shall show how we can realize \(\EM\) in this semantics.

Interactive realizability aims at producing a realizer of the \(\EM\) axiom, a weakened form of the excluded middle restricted to \(\Sigma^0_1\) formulas. 
A generic instance of \(\EM\) is written as:
\[ 
  \EM(\lafa, t_1, \dotsc, t_k) \equiv (\qforall\lvb \lafa(t_1, \dotsc, t_k, \lvb)) \lor (\qexists\lvb \lnot \lafa(t_1, \dotsc, t_k, \lvb)). 
\]
for any \(k+1\)-ary relation \(\lafa\) and arithmetic terms \(t_1, \dotsc, t_k\). 
We call \emph{universal} (resp. \emph{existential}) \emph{disjunct} the first (resp. the second) disjunct of \(\EM(\lafa, t_1, \dotsc, t_k)\). 
For more information on \(\EM\) see \cite{akamaBHK04}. 

The main hurdle we have to overcome in order to build a realizer of \( \EM(\lafa, t_1, \dotsc, t_k) \) is that, by the well-known undecidability of the halting problem, there is no total recursive function that can choose which one of the disjuncts holds. 
Moreover, if the realizer chooses the existential disjunct, it should also be able to provide a witness. 

As we said before terms of type \(\State\) contain knowledge about witnesses of \(\Sigma^0_1\) formulas. 
In order to query a state \(s\) for a witness \(n\) of 
\(\qexists\lvb \lafa(\num n_1, \dotsc, \num n_k, \lvb)\) for some natural numbers \(n_1, \dotsc, n_k\), we need to extend system \(T'\) with the family of term constants:
\[ \query_\lafa : \State \tarrow \underbrace{\Nat \tarrow \dotsb \tarrow \Nat}_k \tarrow \Unit + \Nat. \] 
indexed by \(\lafa \in \LRel_{k+1} \) (and implicitly by \(k \ge 0\)). 
The value of \( \query_\lafa s \num n_1 \dotsm \num n_k \) should be either \(\unit\) if the \(s\) contains no information about such an \(n\) or a numeral \(\num n\) such that \( \interp\lafa(n_1, \dotsc, n_k, n) \) is true. 
More formally we require that \(\query_\lafa\) satisfies the following syntactic property: 
\begin{equation} \tag{IR1} \label[property]{query_prop}
  \query_\lafa s \num n_1 \dotsm \num n_k \leadsto \inr \num n \text{ entails that } \lafa (\num n_1, \dotsc, \num n_k, \num n) \text{ holds}
\end{equation} 
for all natural numbers \( n_1, \dotsc, n_k \). 
This amounts to require that state do not answer with wrong witnesses and it follows immediately from the intended interpretation if we suitably define \( \query_\lafa s \num n_1 \dotsm \num n_k \) using \( \interp s (\lafa, (n_1, \dotsc, n_k)) \). 

An interactive realizer \(\mon\rr_\lafa\) of \(\EM(\lafa)\) will behave as follows. 
When it needs to choose one of the disjuncts it queries the state. 
If the state answer with a witness, \(\mon\rr_\lafa\) reduces to a realizer \(\mon\rr_\exists\) of the existential disjunct containing the witness given by the state. 
Otherwise we can only assume (since we do not know any witness) that the universal disjunct holds and thus \(\mon\rr_\lafa\) reduces to a realizer \(\mon\rr_\forall\) of the universal disjunct. 
This assumption may be wrong if the state is not big enough. 
When \(\mon\rr_\forall\) is evaluated on numerals (this correspond to the fact that an instance \(\lafa(\num n_1, \dotsc, \num n_k, \num n)\) of the universal disjunct assumption is used in the proof), \(\mon\rr_\forall\) checks whether the instance holds. 
If this is not the case the realizer made a wrong assumption and \(\mon\rr_\forall\) reduces to an exceptional value, with the effect of halting the regular reduction and returning the exceptional value. 
For this we need to extend the system \(T'\) with the last family of terms:
\[ \eval_\lafa : \underbrace{\Nat \tarrow \dotsb \tarrow \Nat}_k \tarrow \Nat \tarrow \Unit + \Ex, \] 
again indexed by \(\lafa \in \LRel_k\). 
We shall need \(\eval_\lafa\) to satisfy the following property: 
\[ \tag{IR2} \label[property]{eval_prop}
  \eval_\lafa \num n_1 \dotsm \num n_k \num n \leadsto \inl \unit \text{ entails that } \lafa (\num n_1, \dotsc, \num n_k, \num n) \text{ does not hold}, 
\] 
for all natural numbers \(n_1, \dotsc, n_k, n \). 
This guarantees that if the universal disjunct instance does not hold \(\eval_\lafa\) reduces to an exceptional value. 
Thus an interactive realizer which uses a false instance of an universal assumption cannot reduce to a regular value. 

The last property we need is that for any state \(s\) and natural numbers \(n_1, \dotsc, n_k\), 
\[ \tag{IR3} \label[property]{query_eval_prop} 
  \left. \begin{array}{r}
    \query_\lafa s \num n_1 \dotsm \num n_k \leadsto \inl \unit \\
    \quad \eval \num n_1 \dotsm \num n_k \leadsto \inr e 
\end{array} \right\}
\text{ entails that \(e\) properly extends \(s\).} 
\]
This condition guarantees that we have no ``lazy'' realizers that throw exceptions encoding witnesses that are already in the state. 


\newtermconstant\emReal{\termname{em}}
Now we can define a realizer for \(\EM(\lafa, t_1, \dotsc, t_k)\) as follows:
\begin{align*}
  \mon\emReal(\lafa, t_1, \dotsc, t_k)
  \equiv \abstr{s}\State \inl 
  (\case &(\query_\lafa s t_1 \dotsm t_k) \\
         &(\abstr\_\Unit \inl (\abstr\lvb\Nat \abstr\_\State \eval_\lafa t_1 \dotsm t_k \lvb)) \\
         &(\abstr\lvb\Nat \inr (\pair \lvb \monUnit))). 
\end{align*} 
Of course we need to check that our definition is correct. 
\begin{lemma}[Interactive Realizer for \(\EM\)] \label{thm:realizer_for_em}
  Given any \(\EM\) instance \(\EM(\lafa, t_1, \dotsc, t_k)\), the decorated sequent:
  \begin{equation} \label{eq:em_seq} 
    \alpha_1 : \fa_1, \dotsc, \alpha_l : \fa_l \monSeq \mon\emReal(\lafa, t_1, \dotsc, t_k) : \EM(\lafa, t_1, \dotsc, t_k), 
  \end{equation}
  is valid with respect to the interactive realizability semantics given in \cref{def:interactive_realizability_semantics}. 
\end{lemma}

Then we can extend our proof decoration for \(\HA\) (see \cref{fig:monadic_decorated_rules}) with the new axiom rule: 
\[ 
  \PrAx{}
  \PrLbl\EM 
  \PrUn{\Gamma \monSeq \mon\emReal(\lafa, t_1, \dotsc, t_k) : \EM(\lafa, t_1, \dotsc, t_k)}
  \DisplayProof 
\]
and show that interactive realizability realizes the whole \(\HA+\EM\). 
\begin{theorem}[Soundness of \(\HA+\EM\) with respect to Interactive Realizability Semantics] \label{thm:em_soundness}
  Let \(\mathcal{D}\) be a derivation of \( \Gamma \seq \fa \) in \(\HA+\EM\). Then \( \Gamma \monSeq \mathcal{D}^* : \fa \), where \(\mathcal{D}^*\) is the term obtained by decorating \(\mathcal{D}\), is valid with respect to the interactive realizability semantics. 
\end{theorem} 

\section{Conclusions} 

As we mentioned in the introduction, 
interactive realizability describes a learning by trial-and-error process. 
In our presentation we focused on the evaluation of interactive realizers, which corresponds to the trial-and-error part and is but a single step in the learning process.  
For the sake of completeness, we briefly describe the learning process itself. 

We can interpret an interactive realizer \(\mon\rr\) of a formula \(\fa\) as a function \(f\) from states to states. 
Recall that the intended interpretation of a term \(e : \Ex \) is a function that extends states. 
Then we can define \(f\) by means of \(\mon\rr\) as follows:
\[
  f (s) = \begin{cases}
    \interp{e} (s) \qquad & \text{if } \mon\rr \leadsto \inr e, \\ 
    s \qquad & \text{if } \mon\rr \leadsto \inl t \text{ for some t.}
    \end{cases}
\]
Note that by definition of \(\monRe\) we know that in the first case \(\interp{e} s\) properly extends \(s\). 
We can think of \(f\) as a learning function: we start from a knowledge state and try to prove \(\fa\) with \(\mon\rr\). If we fail, we learn some information that was not present in the state and we use it to extend the state. If we succeed then we do not learn anything and we return the input state. 
Thus note that the fixed points of \(f\) are exactly the states containing enough information to prove \(\fa\).

By composing \(f\) with itself we obtain a learning process: we start from some state (for instance the empty one) and we apply \(f\) repeatedly. 
If in this repeated application eventually produces a fixed point, the learning process ends, since we have the required information to prove \(\fa\). 
Otherwise we build an infinite sequence of ever increasing knowledge states whose information is never enough to prove \(\fa\). 
The fact that the learning process described by interactive realizability ends is proved in Theorem 2.15 of \cite{aschieriB10}. 

In order to express the learning process in system \(T'\) we would need some sort of fix point operator. 
However, we do not need control operators or even the continuation monad, since we simulate exceptions by means of the exception monad without really interrupting the evaluation of our realizers. 
Unfortunately the price for this simplicity is that the learning process is inefficient: each time a realizer reduces to an exceptional value, we start again its evaluation from the beginning, even though the initial part of the evaluation remains the same. 

We wish to point out one of the main differences between our presentation of interactive realizability and the one given in \cite{aschieriB10}. 
In \cite{aschieriB10}, the formula-as-types correspondence is closer to the standard one. 
Exceptions are allowed only at the level of atomic formulas and \(\exmerge\) is only used in atomic rules. 
For instance a realizer for a conjunction \(\fa \land \fb\) could normalize to \(\pair{e_1}{e_2}\). 
In this case, the failure of the realizer is not apparent (at least at the top level) and it is not clear which one of \(e_1\) or \(e_2\) we are supposed to extend the state with. 
In our version exceptions are allowed at the top level of any formula and they ``climb'' upwards whenever possible by means of \(\exmerge\). 


\bibliographystyle{plain}
\bibliography{realizability}

\newpage
\appendix
\section{Technical Appendix}
In this section we collect the parts that did not fit in the page limit. 

\setmonad\MM

\subsection{Omitted Remarks}

\begin{remark}
  The definition of \(\mm\cdot\) and \(\m\cdot\) can be derived from the Curry-Howard correspondence between formulas and types and from a call-by-name monadic translation for types. 
  \newcommand\formType[1]{\lvert #1 \rvert}
  We define the standard interpretation \(\formType\cdot\) that maps a formula into the type of its realizers:
  \begin{align*}
    \formType\lafa &= \Unit, & 
    \formType{\fa \land \fb} &= \formType\fa \times \formType\fb, \\ 
    \formType{\fa \lor \fb} &= \formType\fa + \formType\fb, & 
    \formType{\fa \limply \fb} &= \formType\fa \tarrow \formType\fb,\\ 
    \formType{\qforall{x} \fa} &= \Nat \tarrow \formType\fa, & 
    \formType{\qexists{x} \fa} &= \Nat \times \formType\fa. 
  \end{align*}
  Next we define a translation \(\monTrans\cdot\) that lifts types to their monadic counterparts:
  \begin{align*}
    \monTrans{\ta_0} &\equiv \ta_0, &
    \monTrans{\ta\tarrow\tb} &\equiv \monTrans\ta \tarrow \T\monTrans\tb, \\
    \monTrans{\ta\times\tb} &\equiv \monTrans\ta \times \monTrans\tb, & 
    \monTrans{\ta+\tb} &\equiv \monTrans\ta + \monTrans\tb,  
  \end{align*}
  where \(\ta_0\) is a ground type. 
  The first two clauses are taken from \cite{wadler94} and the other ones are a simple extension, based on the idea that products and sums behave like ground types. 

  By composition we can define the types for the monadic realizers of a formula: 
  \[
    \m\fa \equiv \monTrans{\formType\fa}, \qquad
    \mm\fa \equiv \T\m\fa.
  \]
  Expanding the definitions we get :
  \begin{align*}
    \m\lafa &= \Unit, \\
    \m{\fa\land\fb} &= \monTrans{\formType\fa} \times \monTrans{\formType\fb} = \m\fa \times \m\fb, \\
    \m{\fa\lor\fb} &= \monTrans{\formType\fa} + \monTrans{\formType\fb} = \m\fa + \m\fb, \\
    \m{\fa\limply\fb} &= \monTrans{\formType\fa} \tarrow \T\monTrans{\formType\fb} = \m\fa \tarrow \T\m\fb, \\ 
    \m{\qforall{x} \fa} &= \monTrans\Nat \tarrow \T\monTrans{\formType\fb} = \Nat \tarrow \T\m\fa, \\ 
    \m{\qexists{x} \fa} &= \monTrans\Nat \times \monTrans{\formType\fb} = \Nat \times \m\fa. 
  \end{align*}
  This is the same translation we described in \cref{def:mon_types}. 
\end{remark}

A slightly longer example of syntactic monad. 
\begin{example} 
  \setmonad\ExM 
  A simple but non-trivial example is the exception monad \(\ExM\).
  It describes computations which may either succeed and yield a (normal) value or fail and yield a description of the failure. 
  Consider the usual predecessor function \(\termname{pred}: \Nat \tarrow \Nat\) on the natural numbers: 
  since zero has no predecessor it is common to define \(\termname{pred} \num 0\) as zero. 
  Instead with \(\ExM\) we could have \(\termname{pred} \num 0\) fail and yield a string\footnote{assuming we had strings in our calculus} saying ``zero has no predecessor''. 

  Let \(\Ex\) be a new ground type and let \(\exmerge : \Ex \tarrow \Ex \tarrow \Ex\) be a new constant term. 
  We think terms of type \(\Ex\) as descriptions of failures and we call them \emph{exceptions}. 
  We think of \(\exmerge\) as an operation that merges the information of multiple exceptions when there are multiple failures in a computations. 
  Now we can define the syntactic monad \(\ExM\) as:
  \begin{align*}
    \T\ta &\equiv \ta + \Ex, \qquad 
    \monUnit[\ta] \equiv \abstr\tta\ta \inl[\ta,\Ex]\tta, \\
    \monStar[\ta,\tb] &\equiv 
    \abstr{f}{\ta \tarrow \tb + \Ex} \abstr\tta\ta \case[\ta,\Ex,\tb+\Ex] \tta f \inr[\tb,\Ex], \\  
    \monMerge[\ta,\tb] &\equiv 
    \abstr{\mon\tta}{\ta+\Ex} \abstr{\mon\ttb}{\tb+\Ex} \case[\ta,\Ex,(\ta\times\tb) + \Ex] \mon\tta \\
    &\phantomrel\equiv (\abstr\tta\ta \case[\tb,\Ex,(\ta\times\tb)+\Ex] \mon\ttb 
    (\abstr\ttb\tb \inl[\ta\times\tb,\Ex] (\pair[\ta,\tb] \tta \ttb)) \inr[\ta\times\tb,\Ex] ) \\
    &\phantomrel\equiv (\abstr{e_1}\Ex \case[\tb,\Ex,(\ta\times\tb)+\Ex] \mon\ttb 
    (\abstr\ttb\tb \inr[\ta\times\tb,\Ex] e_1) (\abstr{e_2}\Ex \inr[\ta\times\tb,\Ex] \exmerge e_1 e_2)) . 
  \end{align*}
\end{example}
 
\begin{remark}\allowdisplaybreaks[1]
  In \cref{fig:monadic_decorated_rules}, we wrote all realizers using only \(\monRaiseN k\) and \(\monStarN k\) for the sake of consistency, 
  but note that \(\monRaiseN 0\) could have been replaced by \(\monUnit\) since it reduces to it:
  \begin{align*}
    \monRaiseN{0}
    &\equiv_{\monRaiseN{0}} \abstr{f}\tc \monStarN{0} (\monUnit f) \\ 
    & \equiv_{\monStarN{0}} \abstr{f}\tc (\abstr{f}{\T\tc} f) (\monUnit f) \\ 
    & \reducesto_\beta \abstr{f}\tc \monUnit f \\ 
    & =_\eta \monUnit 
  \end{align*}
  Moreover \(\monRaiseN 2 \pair\) reduces to \(\monMerge\):
  \begin{align*} 
    \monRaiseN{2} \pair &\equiv_{\monRaiseN 2}
    (\abstr{f}{\ta \tarrow \tb \tarrow \ta \times \tb} \monStarN{2} (\abstr\lva\ta \abstr\lvb\tb \monUnit (f \lva \lvb))) \pair \\ 
    &\leadsto_\beta \monStarN{2} (\abstr\lva\ta \abstr\lvb\tb \monUnit (\pair \lva \lvb)) \\ 
    &\equiv_{\monStarN 2} (\abstr{f}{\ta \tarrow \tb \tarrow \T(\ta \times \tb)} \abstr\lva{\T\ta} \abstr\lvb{\T\tb}  \\
    &\phantomrel{\equiv_{\monStarN 2}}
    \monStarN{k} (\abstr\lvc{\ta \times \tb} f (\prl \lvc) (\prr \lvc) ) (\monMerge \lva \lvb)) (\abstr\lva\ta \abstr\lvb\tb \monUnit \pair\lva\lvb) \\ 
    &\leadsto_\beta \abstr\lva{\T\ta} \abstr\lvb{\T\tb} \monStar (\abstr\lvc{\ta \times \tb} (\abstr\lva\ta \abstr\lvb\tb \monUnit \pair\lva\lvb) (\prl \lvc) (\prr \lvc) ) (\monMerge \lva \lvb) \\ 
    &\leadsto_\beta \abstr\lva{\T\ta} \abstr\lvb{\T\tb} \monStar (\abstr\lvc{\ta \times \tb} \monUnit \pair{\prl \lvc}{\prr \lvc}) (\monMerge \lva \lvb) \\ 
    &=_\times \abstr\lva{\T\ta} \abstr\lvb{\T\tb} \monStar (\abstr\lvc{\ta \times \tb} \monUnit \lvc) (\monMerge \lva \lvb) \\ 
    &=_\eta \abstr\lva{\T\ta} \abstr\lvb{\T\tb} \monStar \monUnit (\monMerge \lva \lvb) \\ 
    &\leadsto_{\ref{mon2}} \abstr\lva{\T\ta} \abstr\lvb{\T\tb} \monMerge \lva \lvb \\ 
    &=_\eta \monMerge, 
  \end{align*} 
  so we could replace it in \(\RuleNameI\land\). 
\end{remark}

\subsection{Proofs Omitted from \Cref{sec:monadic_realizability}}
Here we collect the proofs that we omitted. 

In order to prove \cref{thm:ha_soundness}, we need to show that \(\monStarN{k}\) and \(\monRaiseN{k}\) satisfy a generalization of \cref{real2}. 
\begin{lemma}[Monadic Realizability Property for \(\monStarN{k}\)] \label{thm:starN} 
  Let \( \fa_1, \dotsc, \fa_k \) and \(\fb\) be any formulas and let \( \rr : \m{\fa_1} \tarrow \dotsb \tarrow \m{\fa_k} \tarrow \mm\fb \) be a term. 
  Assume that, for all terms \( \ra_1 : \m{\fa_1}, \dotsc, \ra_k : \m{\fa_k} \) such that \( \ra_1 \re \fa_1, \dotsc, \ra_k \re \fa_k \), we have: 
  \[ \rr \ra_1 \dotsm \ra_k \monRe \fb. \] 
  Then, for all terms \( \mon\ra_1 : \mm{\fa_1}, \dotsc, \mon\ra_k : \mm{\fa_k} \) such that \( \mon\ra_1 \monRe \fa_1, \dotsc, \mon\ra_k \monRe \fa_k \), we have: 
  \[ \monStarN{k} \rr \mon\ra_1 \dotsm \mon\ra_k \monRe \fb. \] 
\end{lemma}
\begin{proof}
  By induction on \(k\). 
  For \( k = 0 \) it is trivial and for \( k = 1 \) it follows from \cref{real2} since \( \monStarN{1} \equiv \monStar \). 
  Now we just need to prove that if the statement holds for some \( k \ge 1 \), it holds for \( k+1 \) too. 

  As in the statement we assume that, for all terms \( \ra_1 : \m{\fa_1}, \dotsc, \ra_{k+1} : \m{\fa_{k+1}} \) such that \( \ra_1 \re \fa_1, \dotsc, \ra_{k+1} \re \fa_{k+1} \): 
  \[ \rr \ra_1 \dotsm \ra_{k+1} \monRe \fb, \] 
  and that \( \mon\ra_1 : \mm{\fa_1}, \dotsc, \mon\ra_{k+1} : \mm{\fa_{k+1}} \) are terms such that \( \mon\ra_1 \monRe \fa_1, \dotsc, \mon\ra_{k+1} \monRe \fa_{k+1} \). 
  We need to show that: 
  \[ \monStarN{k+1} \rr \mon\ra_1 \dotsm \mon\ra_{k+1} \monRe \fb. \] 
  Since we know by definition of \(\monStarN{k+1}\) that \( \monStarN{k+1} \rr \mon\ra_1 \dotsm \mon\ra_{k+1} \) reduces to the term: 
  \[ \monStarN{k} (\abstr{z}{\m{\fa_1} \times \m{\fa_2}} \rr (\prl z) (\prr z) ) (\monMerge \mon\ra_1 \mon\ra_2) \mon\ra_3 \dotsm \mon\ra_{k+1}, \]
  and by \cref{real3} that \( \monMerge \mon\ra_1 \mon\ra_2 \monRe \fa_1 \land \fa_2 \),
  we see that we can use the inductive hypothesis on \(k\) to conclude. 
  In order to do so we have to show that the assumption of the inductive hypothesis holds, namely that, for any \( \ra_1 : \m{\fa_1} \times \m{\fa_2} \), \( \ra_3 : \m{\fa_3}, \dotsc, \ra_k : \m{\fa_k} \) such that \( \ra_1 \re \fa_1 \land \fa_2 \), \( \ra_2 \re \fa_2, \dotsc, \ra_k \re \fa_k \) it is the case that: 
  \[ (\abstr{z}{\m{\fa_1} \times \m{\fa_2}} \rr (\prl z) (\prr z) ) \ra_1 \dotsm \ra_k \monRe \fb. \] 
  By reducing the realizer we get that this is equivalent to: 
  \[ \rr (\prl \ra_1) (\prr \ra_1) \ra_2 \dotsm \ra_k \monRe \fb, \] 
  which is true by the assumption on \(\rr\) since \(\ra_1 \re \fa_1 \land \fa_2 \) means that \( \prl \ra_1 \re \fa_1 \) and \( \prr \ra_1 \re \fa_2 \) by definition of \(\re\). 
\end{proof}

We prove a similar property for \(\monRaiseN{k}\). 
\begin{lemma}[Monadic Realizability Property for \(\monRaiseN{k}\)] \label{thm:raiseN} 
  Let \( \fa_1, \dotsc, \fa_k \) and \(\fb\) be any formulas and let \( \rr : \m{\fa_1} \tarrow \dotsb \tarrow \m{\fa_k} \tarrow \m\fb \) be a term. 
  Assume that, for all terms \( \ra_1 : \m{\fa_1}, \dotsc, \ra_k : \m{\fa_k} \) such that \( \ra_1 \re \fa_1, \dotsc, \ra_k \re \fa_k \), it is the case that: 
  \[ \rr \ra_1 \dotsm \ra_k \re \fb. \] 
  Then, for all terms \( \mon\ra_1 : \mm{\fa_1}, \dotsc, \mon\ra_k : \mm{\fa_k} \) such that \( \mon\ra_1 \monRe \fa_1, \dotsc, \mon\ra_k \monRe \fa_k \), we have that: 
  \[ \monRaiseN{k} \rr \mon\ra_1 \dotsm \mon\ra_k \monRe \fb. \] 
\end{lemma}
\begin{proof}
  Assume that, for all terms \( \ra_1 : \m{\fa_1}, \dotsc, \ra_k : \m{\fa_k} \) such that \( \ra_1 \re \fa_1, \dotsc, \ra_k \re \fa_k \), it is the case that: 
  \[ \rr \ra_1 \dotsm \ra_k \re \fb, \] 
  and let \( \mon\ra_1 : \mm{\fa_1}, \dotsc, \mon\ra_k : \mm{\fa_k} \) be terms such that \( \mon\ra_1 \monRe \fa_1, \dotsc, \mon\ra_k \monRe \fa_k \). 
  We want to prove that: 
  \[ \monRaiseN{k} \rr \mon\ra_1 \dotsm \mon\ra_k \monRe \fb. \] 
  By definition of \(\monRaiseN{k}\) this reduces to: 
  \[ \monStarN{k} (\abstr{x_1}{\m{\fa_1}} \dotsm \abstr{x_k}{\m{A_k}} \monUnit (\rr x_1 \dotsm x_k)) \mon\ra_1 \dotsm \mon\ra_k \monRe \fb. \] 
  This follows by \cref{thm:starN} if we can show that, for any \( \ra_1 : \m{\fa_1}, \dotsc, \ra_k : \m{\fa_k} \) such that \( \ra_1 \re \fa_1, \dotsc, \ra_k \re \fa_k \), we have: 
  \[ (\abstr{x_1}{\m{\fa_1}} \dotsm \abstr{x_k}{\m{\fa_k}} \monUnit (\rr x_1 \dotsm x_k)) \ra_1 \dotsm \ra_k \monRe \fb. \]
  Reducing the realizer we get that this is equivalent to: 
  \[ \monUnit (\rr \ra_1 \dotsm \ra_k) \monRe \fb, \]
  and this follows by \cref{real1} and by assumption on \(\rr\). 
\end{proof}

\begin{proof}[Proof of \Cref{thm:ha_soundness}]
  We proceed by induction on the structure of the decorated version of \(\mathcal{D}\), that is, we assume that the statement holds for all decorated sub-derivations of \(\mathcal{D}\) and we prove that it holds for \(\mathcal{D}\) too. 
  More precisely we have to check the soundness of each decorated rule, showing that the validity of the premises yields the validity of the conclusion. 

  We start with some general notation and observations. 
  Let \( \Gamma \equiv \alpha_1 : \fa_1, \dotsc, \alpha_k : \fa_k \) for some \(k\). 
  Following the notation in \cref{def:mon_sem}, 
  we fix natural numbers \( n_1, \dotsc, n_l \) and terms \( \rr_1 : \fa_1, \dotsc, \rr_k : \fa_k \), 
  we define abbreviations: 
  \begin{align*}
    \Omega &\equiv x_1 \coloneqq \num n_1, \dotsc, x_l \coloneqq \num n_l, \\ 
    \Sigma &\equiv \alpha_1 \coloneqq r_1, \dotsc, \alpha_k \coloneqq r_k, 
  \end{align*}
  and we assume that: 
  \[ \rr_1 \re \fa_1[\Omega] \qquad \dotso \qquad \rr_k \re \fa_k[\Omega]. \]

  Note that if some term \(t : \ta_1 \tarrow \dotsb \tarrow \ta_k \tarrow \tb \) has no free variables then \( (t a_1 \dotsm a_k)[\Omega, \Sigma] \equiv t (a_1[\Omega, \Sigma]) \dotsm (a_k[\Omega, \Sigma]) \). 
  In particular this holds if \(t\) is one of \(\monStarN{k}\), \(\monRaiseN{k}\), \(\pair\), \(\prl\), \(\prr\), \(\case\), \(\inl\), \(\inr\). 
  The same holds for formulas, so \( (\fa \star \fb)[\Omega] \equiv \fa[\Omega] \star \fb[\Omega] \) where \(\star\) is one of \( \land, \lor \) or \( \limply \). 
  Also note that \( \m{\fa[\Omega]} = \m \fa \) since \( \m \cdot \) does not depend on the terms in \(\fa\). 
  In particular the types of the proof variables in \(\Gamma\) do not change, meaning we do not need to perform substitutions in \(\Gamma\).  
  We shall take advantage of these facts without mentioning it. 

  Now we can start showing that the rules are sound. \fixme{is the order of \(\Omega\) and \(\Sigma\) relevant? do we need to substitute \(\Omega\) in realizers?}

  \begin{itemize}
    \item[Id]
      We have to prove that: 
      \[ (\monRaiseN{0} \alpha_i)[\Omega, \Sigma] \monRe \fa[\Omega], \]
      where \( \fa = \fa_i \) for some \( i \in \{1, \dotsc, k\} \). 

      By performing the substitutions, we can rewrite the realizer as \(\monRaiseN{0} \rr_i\) so we need to prove that: 
      \[ \monRaiseN{0} \rr_i \monRe \fa. \]
      This follows by \cref{thm:raiseN} since by assumption \( \rr_i \re \fa_i[\Omega] \).

    \item[Atm]
      We have to prove that: \fixme{\(l\) is already used in this theorem}
      \[ (\monRaiseN{l} (\abstr{\gamma_1}\Unit \dotsm \abstr{\gamma_l}\Unit \unit) \mon\rr_1 \dotsm \mon\rr_l)[\Omega,\Sigma] \monRe \lafa[\Omega]. \]
      By performing the substitutions, we can rewrite the realizer as: 
      \[ \monRaiseN{l} (\abstr{\gamma_1}\Unit \dotsm \abstr{\gamma_l}\Unit \unit) \mon\rr_1[\Omega,\Sigma] \dotsm \mon\rr_l[\Omega,\Sigma]. \]
      By inductive hypothesis we know that 
      \[ \mon\rr_1[\Omega,\Sigma] \monRe \lafa_1[\Omega], \dotsc, \mon\rr_l[\Omega,\Sigma] \monRe \lafa_l[\Omega], \] 
      and thus we can conclude by \cref{thm:raiseN} if we can show that: 
      \[ (\abstr{\gamma_1}\Unit \dotsm \abstr{\gamma_l}\Unit \unit) \rr_1 \dotsm \rr_l \re \lafa[\Omega], \]
      for all \(\rr_1, \dotsc, \rr_l\) that are inner realizers of \(\lafa_1, \dotsc, \lafa_l\) respectively. 
      Since 
      \[ (\abstr{\gamma_1}\Unit \dotsm \abstr{\gamma_l}\Unit \unit) \rr_1 \dotsm \rr_l, \]
      reduces to \(\unit\) and \(\unit \re \lafa[\Omega]\) by definition of \(\re\) we are done. 
  \end{itemize}
  In the following we will apply the substitutions directly without mentioning it. 
  \begin{itemize}
    \item[\(\RuleNameI\land\)] 
      We have to prove that 
      \[ \monRaiseN{2} \pair \mon\ra[\Omega,\Sigma] \mon\rb[\Omega,\Sigma] \monRe \fa[\Omega] \land \fb[\Omega], \] 
      assuming that \( \mon\ra[\Omega, \Sigma] \monRe \fa[\Omega] \) and \( \mon\rb[\Omega, \Sigma] \monRe \fa[\Omega]\). 
      This follows by \cref{thm:raiseN} since
      \[ \pair \ra \rb \re \fa \land \fb, \] 
      for all inner realizers \(\ra\) of \(\fa\) and \(\rb\) of \(\fb\), 
      by definition of \(\re\). 

    \item[\( {\RuleNameE[L]\land} \)]
      We have to prove that 
      \[ (\monRaiseN{1} \prl \mon\rr)[\Omega,\Sigma] \monRe \fa[\Omega], \] 
      assuming that 
      \[ \mon\rr[\Omega,\Sigma] \monRe \fa[\Omega] \land \fb[\Omega]. \] 
      This follows by \cref{thm:raiseN} if 
      \[ \prl \rr \re \fa[\Omega], \] 
      for any inner realizer \(\rr\) of \(\fa[\Omega] \land \fb[\Omega]\).
      This is the case because from \( \rr \re \fa \land \fb \) if and only if \( \prl \rr \re \fa \) by definition of \(\re\). 

    \item[\( {\RuleNameE[R]\land} \)]
      Very similar to the proof for \(\RuleName[L]\land{E}\). 

    \item[\( {\RuleNameI[L]\lor} \)]
      We have to show that: 
      \[ \monRaiseN{1} \inl \mon\ra[\Sigma,\Omega] \monRe \fa[\Omega] \lor \fb[\Omega], \] 
      assuming that: 
      \[ \mon \ra[\Sigma,\Omega] \monRe \fa[\Omega]. \] 
      This follows by \cref{thm:raiseN} if 
      \[ \inl \ra \re \fa[\Omega], \] 
      for any inner realizer \(\ra\) of \(\fa[\Omega]\). 
      This is the case since \( \ra \re \fa[\Omega] \) if and only if \( \inl \ra \re \fa[\Omega] \lor \fb[\Omega] \) by definition of \(\re\). 

    \item[\( {\RuleNameI[R]\lor} \)]
      Very similar to the proof for \(\RuleName[L]\lor{I}\). 

    \item[\( \RuleNameE\lor \)] 
      We have to show that: 
      \[ \monStarN{1} (\abstr\gamma{\m\fa + \m\fb} \case \gamma (\abstr\alpha{\m\fa} \mon\ra[\Omega,\Sigma]) (\abstr\beta{\m\fb} \mon\rb[\Omega,\Sigma])) \mon\rr[\Omega,\Sigma] \monRe \fc[\Omega] \] 
      assuming by inductive hypothesis that: 
      \fixme{\(\alpha\) and \(\beta\) are inconsistent notation}
      \begin{enumerate} 
        \item \( \mon\rr[\Omega,\Sigma] \monRe \fa[\Omega] \lor \fb[\Omega] \), 
        \item \( \mon\ra[\Omega,\Sigma,\alpha \coloneqq \ra]\ \monRe \fc[\Omega] \) for any inner realizer \(\ra\) of \(\fa[\Omega]\), 
        \item \( \mon\rb[\Omega,\Sigma,\beta \coloneqq \rb]\ \monRe \fc[\Omega] \) for any inner realizer \(\rb\) of \(\fb[\Omega]\). 
      \end{enumerate} 
      We can conclude by \cref{thm:starN} if we show that 
      \[
        (\abstr\gamma{\m\fa + \m\fb} \case \gamma (\abstr\alpha{\m\fa} \mon\ra[\Omega,\Sigma]) (\abstr\beta{\m\fb} \mon\rb[\Omega,\Sigma])) \rr,  
      \]
      which \(\beta\)-reduces to
      \begin{equation}\label{eq:or_elim}
        \case \rr (\abstr\alpha{\m\fa} \mon\ra[\Omega,\Sigma]) (\abstr\beta{\m\fb} \mon\rb[\Omega,\Sigma]),  
      \end{equation}
      is a monadic realizer of \(\fc[\Omega]\) 
      for any inner realizer \(\rr\) of \(\fa[\Omega] \lor \fb[\Omega]\). 

      By definition of \(\re\), we know that either 
      \(\rr \leadsto \inl\ra\) where \(\ra\) is an inner realizer of \(\fa[\Omega]\) or 
      \(\rr \leadsto \inr\rb\) where \(\rb\) is an inner realizer of \(\fb[\Omega]\). 
      Assume that we are in the first case (the second case is analogous). 
      Then \eqref{eq:or_elim} becomes: 
      \[ \case (\inl\ra) (\abstr\alpha{\m\fa} \mon\ra[\Omega,\Sigma]) (\abstr\beta{\m\fb} \mon\rb[\Omega,\Sigma]), \]
      which reduces to 
      \[ (\abstr\alpha{\m\fa} \mon\ra[\Omega,\Sigma]) \ra, \]
      and to
      \[ \mon\ra[\Omega,\Sigma,\alpha \coloneqq \ra], \]
      which is a monadic realizer of \(\fc[\Omega]\) by inductive hypothesis. 

    \item[\( \RuleNameI\limply \)]
      We have to show that: 
      \[ \monRaiseN{0} (\abstr{\alpha_{k+1}}{\m\fa} \mon\rr[\Omega,\Sigma]) \monRe \fa[\Omega] \limply \fb[\Omega], \]
      assuming that: 
      \[ \mon \rr[\Omega,\Sigma,\alpha_{k+1} \coloneqq \ra] \monRe \fb[\Omega], \]
      for any inner realizer \(\ra\) of \(\fa[\Omega] \). 
      By \cref{thm:raiseN} it is enough to show that: 
      \[ \abstr{\alpha_{k+1}}{\m\fa} \mon\rr[\Omega,\Sigma] \re \fa[\Omega] \limply \fb[\Omega]. \]
      By definition of \(\re\) this holds if and only if: 
      \[ (\abstr{\alpha_{k+1}}{\m \fa} \mon\rr[\Omega,\Sigma]) \ra \monRe \fb[\Omega], \]
      for any inner realizer \(\ra\) of \(\fa[\Omega] \). 
      Reducing we get: 
      \[ \mon\rr[\Omega,\Sigma][\alpha_{k+1} \coloneqq \ra]) \monRe \fb[\Omega], \]
      and since \( \mon\rr[\Omega,\Sigma][\alpha_{k+1} \coloneqq \ra] \equiv \mon\rr[\Omega,\Sigma,\alpha_{k+1} \coloneqq \ra] \), \fixme{are the substitutions really correct? check this also for or elimination}
      we can conclude by the inductive hypothesis. 

    \item[\( \RuleNameE\limply \)] 
      We have to show that: 
      \[ (\monStarN{2} (\abstr{\gamma_1}{\m\fa \tarrow \m\fb} \abstr{\gamma_2}{\m\fa} \gamma_1 \gamma_2) \mon\rr[\Omega,\Sigma] \mon\ra[\Omega,\Sigma]) \monRe \fb[\Omega], \] 
      assuming by inductive hypothesis that: 
      \begin{enumerate}
        \item \( \mon\rr[\Omega,\Sigma] \monRe \fa[\Omega] \limply \fb[\Omega] \),
        \item \( \mon\ra[\Omega,\Sigma] \monRe \fa[\Omega] \). 
      \end{enumerate}
      This follows by \cref{thm:raiseN} if: 
      \[ (\abstr{\gamma_1}{\m\fa \tarrow \m\fb} \abstr{\gamma_2}{\m\fa} \gamma_1 \gamma_2) \rr \ra, \]
      which \(\beta\)-reduces to 
      \[ \rr \ra, \]
      is a monadic realizer of \( \fb[\Omega]\) 
      for any inner realizers \(\rr\) and \(\ra\) of \(\fa[\Omega] \limply \fb[\Omega]\) and \(\fa[\Omega] \) respectively. 
      This follows immediately by definition of \(\re\). 
  \end{itemize}
  In the following cases we assume that \(\Omega\) does not contain a substitution for the variable \(\lva\)
  and we write it explicitly when it is needed.
  \begin{itemize}
    \item[\( \RuleNameI\forall \)]
      We have to show that: 
      \[ \monRaiseN{0} (\abstr\lva\Nat \mon\rr[\Omega,\Sigma]) \monRe \qforall\lva \fa[\Omega], \]
      assuming by inductive hypothesis that: 
      \[ \mon\rr[\Omega,\lva \coloneqq \num n,\Sigma] \monRe \fa[\Omega,\lva \coloneqq \num n], \]
      for any natural number \(n\). 
      This follows by \cref{thm:raiseN} if: 
      \[ (\abstr\lva\Nat \mon\rr[\Omega,\Sigma]) \re \qforall\lva \fa[\Omega], \]
      which by definition of \(\re\) means that:
      \[ (\abstr\lva\Nat \mon\rr[\Omega,\Sigma]) \num n  \re \fa[\Omega, \lva \coloneqq \num n], \]
      for any natural number \(n\).
      By \(\beta\)-reducing we get:
      \[ \mon\rr[\Omega,\lva \coloneqq \num n,\Sigma]  \re \fa[\Omega, \lva \coloneqq \num n], \]
      which holds by inductive hypothesis. 

    \item[\( \RuleNameE\forall \)]
      We have to show that: 
      \[ (\monStarN{1} (\abstr\gamma{\Nat \tarrow \mm\fa} \gamma (\lta[\Omega]))) \mon\rr[\Omega,\Sigma] \monRe (\fa\subst\lva\lta)[\Omega], \]
      assuming by inductive hypothesis that: 
      \[ \mon\rr[\Omega,\Sigma] \monRe \qforall\lva \fa[\Omega]. \]
      This follows by \cref{thm:starN} if: 
      \[ (\abstr\gamma{\Nat \tarrow \mm\fa} \gamma (\lta[\Omega]))) \rr \leadsto \rr (\lta[\Omega]), \]
      is a monadic realizer of \(\fa[\Omega]\), for any inner realizer \(\rr\) of \(\qforall\lva \fa[\Omega]\).
      This follows by definition of \(\re\) for \(\rr \re \qforall\lva \fa[\Omega]\), 
      since \(\lta[\Omega]\) is closed and thus reduces to a numeral. 
      \fixme{needs normalization of arithmetic terms}

    \item[\( \RuleNameI\exists \)]
      We have to show that: 
      \[ \monRaiseN{1} (\abstr\gamma{\m\fa} \pair\lta[\Omega]\gamma) \mon\rr[\Omega,\Sigma] \monRe \qexists\lva \fa[\Omega], \]
      assuming by inductive hypothesis that: 
      \[ \mon\rr[\Omega,\Sigma] \monRe \fa[\Omega,\lva\coloneqq\lta]. \] 
      This follows by \cref{thm:raiseN} if: 
      \[ (\abstr\gamma{\m\fa} \pair\lta[\Omega]\gamma) \rr \leadsto \pair\lta[\Omega]\rr \]
      is an inner realizer of \( \qexists\lva \fa[\Omega] \), for any inner realizer \(\rr\) of \( \fa[\Omega,\lva\coloneqq\lta] \). 
      This follows by definition of \(\re\). 

    \item[\( \RuleNameE\exists \)] 
      We have to show that: 
      \[ \monStarN{1} (\abstr\gamma{\Nat \times \m\fa} (\abstr\lvb\Nat \abstr\alpha{\m\fa} \mon\rr_2[\Omega,\Sigma])(\prl \gamma)(\prr \gamma)) \mon\rr_1[\Omega,\Sigma] \monRe \fc[\Omega], \]
      assuming by inductive hypothesis that: 
      \begin{enumerate}
        \item \( \mon\rr_1[\Omega,\Sigma] \monRe \qexists\lva \fa[\Omega] \),
        \item \( \mon\rr_2[\Omega,\lvb \coloneqq \num n,\Sigma, \alpha \coloneqq \rr] \monRe \fc[\Omega] \), for any natural number \(n\) and any inner realizer \(\rr\) of \(\fa[\Omega]\).
      \end{enumerate}
      This follows by \cref{thm:starN} and by the inductive hypothesis on \(\mon \rr_1\) if, for any inner realizer \(\rr_1\) of \(\qexists\lva \fa[\Omega]\): 
      \begin{align*}
        &\phantomrel\leadsto (\abstr\gamma{\Nat \times \m\fa} (\abstr\lvb\Nat \abstr\alpha{\m\fa} \mon\rr_2[\Omega,\Sigma])(\prl \gamma)(\prr \gamma)) \rr_1 \leadsto \\ 
        &\leadsto (\abstr\lvb\Nat \abstr\alpha{\m\fa} \mon\rr_2[\Omega,\Sigma])(\prl \rr_1)(\prr \rr_1) \leadsto \\ 
        &\leadsto ((\mon\rr_2[\Omega,\Sigma])\subst\lvb{\prl \rr_1})\subst\alpha{\prr \rr_1} \equiv \\ 
        &\equiv \mon\rr_2[\Omega,\lvb \coloneqq \prl \rr_1,\Sigma, \alpha \coloneqq \prr \rr_1].
      \end{align*}
      is a monadic realizer of \(\fc[\Omega]\).
      By definition of \(\re\) we have that \(\prr \rr_1 \re \fa\subst\lva{\prl \rr_1} \) and thus we can conclude by the inductive hypothesis on \(\mon \rr_2\). 
    \item[Ind] 
      We have to show that: 
      \[ (\monRaiseN{0} (\totRec\infty f))[\Omega, \Sigma] \monRe (\qforall\lva \fa)[\Omega], \] 
      assuming that, for all naturals numbers \(n\) and for all \( \ra : \Nat \tarrow T(\Unit \tarrow T\m\fa) \) such that \( \ra \re \qforall\lvc \lvc < \num n \limply \fa[\lva \coloneqq \lvc] \): 
    \[ \rr[\Omega, \lvb \coloneqq \num n, \Sigma, \alpha_{k+1}] \coloneqq \ra] \monRe \fa[\lva \coloneqq \lvb][\Omega, \lvb \coloneqq \num n]. \] 
    Note that \( \fa[\lva \coloneqq \lvb][\Omega, \lvb \coloneqq \num n] \) is just \( \fa[\Omega, \lva \coloneqq \num n] \). 
    By \cref{thm:raiseN} we get the conclusion if \( \totRec\infty f[\Omega,\Sigma] \re \qforall\lva A[\Omega] \), which by definition of \(\re\) means that 
    \[ \totRec\infty f[\Omega,\Sigma] \num n \monRe A[\Omega, \lva \coloneqq \num n] \] 
    for any natural number \(n\). 
    In order to show this we shall prove that for any natural number \(n\) and any \(\omega \in \N \cup \{\infty\}\) such that either \( \omega = \infty \) or \( \omega > n \), we have:
    \[ \totRec\omega f[\Omega,\Sigma] \num n \monRe A[\Omega, \lva \coloneqq \num n]. \] 
    We proceed by complete induction on \(n\), so we assume that the statement holds for all natural numbers \(m\) such that \( m < n \). 
    We begin by reducing the realizer (in the first step we use the assumption on \(\omega\): 
    \begin{align*} 
      \totRec\omega f[\Omega, \Sigma] \num n 
      &\leadsto f[\Omega, \Sigma] \num n (\totRec{n} f[\Omega, \Sigma]) \\ 
      &\leadsto (\abstr\alpha{} \mon\rr[\Omega, \lvb \coloneqq \num n]) (\abstr\lvc\Nat \monRaiseN{0} (\abstr{\_}\Unit \totRec{n} f[\Omega, \Sigma] \lvc)) \\ 
      &\leadsto \mon\rr[\Omega, \lvb \coloneqq \num n, \Sigma, \alpha \coloneqq \abstr\lvc\Nat \monRaiseN{0} (\abstr{\_}\Unit \totRec{n} f[\Omega, \Sigma] \lvc)] 
    \end{align*} 
    Then we have to show that: 
    \[ \mon\rr[\Omega, \lvb \coloneqq \num n, \Sigma, \alpha \coloneqq \abstr\lvc\Nat \monRaiseN{0} (\abstr{\_}\Unit \totRec{n} f[\Omega, \Sigma] \lvc)] \monRe A[\Omega, \lva \coloneqq \num n]. \] 
    This follows from the inductive hypothesis on the premise of the complete induction rule if we can show that: 
    \[ \abstr\lvc\Nat \monRaiseN{0} (\abstr{\_}\Unit \totRec{n} f[\Omega, \Sigma] \lvc) \re \qforall\lvc \lvc < \num n \limply \fa[\lva \coloneqq \lvc]. \] 
    By definition of \(\re\) this is the case if: 
    \[ \monRaiseN{0} (\abstr{\_}\Unit \totRec{n} f[\Omega, \Sigma] \num m) \monRe \num m < \num n \limply \fa[\lva \coloneqq \num m], \]
    for all  natural numbers \(m\). By \cref{real1} this follows from: 
    \[ \abstr{\_}\Unit \totRec{n} f[\Omega, \Sigma] \num m \re \num m < \num n \limply \fa[\lva \coloneqq \num m]. \]
    Again by definition of \(\re\) this is equivalent to showing that for any \( u : \Unit \) such that \( u \re \num m < \num n \) we have: 
    \[ \totRec{n} f[\Omega, \Sigma] \num m \re \fa[\lva \coloneqq \num m]. \]
    Note that, since \( u : \Unit \), \fixme{by normal form theorem}
    \( u \leadsto \unit \), so there are two possible cases: 
    either \( m < n \) is true and then \( u \re \num m < \re \num n \) for any \( u : \Unit \) or \( m < n \) is false and no \( u : \Unit \) can realize \( \num m < \num n \). 
    In both cases the statement holds: 
    in the former case by inductive hypothesis on \(m\) and in the latter case trivially since the universal quantification on \(u\) is empty. 
\end{itemize}
\end{proof} 

\subsection{Proofs Omitted from \Cref{sec:monadic_interactive_realizability}} 

We did not check that \cref{def:ir_monad} is correct and that \(\IR\) actually is a syntactic monad. 
\begin{lemma}[The Syntactic Monad \(\IR\)] \label{thm:ir_monad}
  \(\IR\) is a syntactic monad. 
\end{lemma} 
\begin{proof}\allowdisplaybreaks[1]
  We just need to check that \(\monUnit\), \(\monStar\) and \(\monMerge\) satisfy all the properties in \cref{def:syntactic_monad}. 
  This amounts to perform some reductions. 
  \begin{itemize}
    \item[\ref{mon1}]
      Given any \(\mon\tta : \T\ta\), we have:
      \begin{align*}
        &\monStar[\ta,\ta] \monUnit[\ta] \mon\tta \\
        &\equiv 
        (\abstr{f}{\ta \tarrow \T\ta} \abstr{\mon \tta}{\T\ta} \abstrS{s} \case[\ta,\Ex,\ta+\Ex](\mon \tta s) (\abstr\tta\ta f \tta s) \inr[\ta,\Ex])
        \monUnit[\ta] \mon\tta \\ 
        &\reducesto_\beta 
        (\abstr{\mon\tta}{\T\ta} \abstrS{s} \case[\ta,\Ex,\ta+\Ex](\mon \tta s) (\abstr\tta\ta \monUnit[\ta] \tta s) \inr[\ta,\Ex])
        \mon\tta \\ 
        &\reducesto_\beta 
        \abstrS{s} \case[\ta,\Ex,\ta+\Ex](\mon\tta s) 
        (\abstr\tta\ta \monUnit[\ta] \tta s) 
        \inr[\ta,\Ex] \\ 
        &\equiv
        \abstrS{s} \case[\ta,\Ex,\ta+\Ex](\mon\tta s) 
        (\abstr\tta\ta (\abstr\tta\ta \abstrS{\_} \inl[\ta,\Ex] \tta) \tta s) 
        \inr[\ta,\Ex] \\ 
        &\reducesto_\beta 
        \abstrS{s} \case[\ta,\Ex,\ta+\Ex](\mon\tta s) 
        (\abstr\tta\ta (\abstrS{\_} \inl[\ta,\Ex] \tta) s) 
        \inr[\ta,\Ex] \\ 
        &\reducesto_\beta 
        \abstrS{s} \case[\ta,\Ex,\ta+\Ex](\mon\tta s) 
        (\abstr\tta\ta \inl[\ta,\Ex] \tta) 
        \inr[\ta,\Ex] \\ 
        &=_\eta
        \abstrS{s} \case[\ta,\Ex,\ta+\Ex](\mon\tta s) 
        \inl[\ta,\Ex] \inr[\ta,\Ex] \\ 
        &=_\times
        \abstrS{s} \mon\tta s \\
        &=_\eta
        \mon\tta, 
      \end{align*}
      as required by \cref{mon1}.
    \item[\ref{mon2}] 
      Given any \( f : \ta \tarrow \T\tb \) and \(\tta : \ta\), we have:
      \begin{align*}
        &\monStar[\ta,\tb] f (\monUnit[\ta] \tta) \\ 
        &\equiv 
        (\abstr{f}{\ta \tarrow \T\tb} \abstr{\mon\tta}{\T\ta} \abstrS{s} \case[\ta,\Ex,\tb+\Ex](\mon\tta s) (\abstr\tta\ta f \tta s) \inr[\tb,\Ex]) f 
        (\monUnit[\ta] \tta) \\ 
        &\reducesto_\beta 
        (\abstr{\mon\tta}{\T\ta} \abstrS{s} \case[\ta,\Ex,\tb+\Ex](\mon\tta s) (\abstr\tta\ta f \tta s) \inr[\tb,\Ex])
        (\monUnit[\ta] \tta) \\ 
        &\reducesto_\beta 
        \abstrS{s} \case[\ta,\Ex,\tb+\Ex] 
        (\monUnit[\ta] \tta s) 
        (\abstr\tta\ta f \tta s) 
        \inr[\tb,\Ex] \\ 
        &\equiv
        \abstrS{s} \case[\ta,\Ex,\tb+\Ex] 
        ((\abstr\tta\ta \abstrS{\_} \inl[\ta,\Ex] \tta) \tta s) 
        (\abstr\tta\ta f \tta s) 
        \inr[\tb,\Ex] \\ 
        &\reducesto_\beta 
        \abstrS{s} \case[\ta,\Ex,\tb+\Ex] 
        (\inl[\ta,\Ex] \tta) 
        (\abstr\tta\ta f \tta s) 
        \inr[\tb,\Ex] \\ 
        &\reducesto_\times 
        \abstrS{s} (\abstr\tta\ta f \tta s) \tta \\
        &\reducesto_\beta 
        \abstrS{s} f \tta s \\
        &=_\eta
        f \tta,
      \end{align*}
      as required by \cref{mon2}.
    \item[\ref{mon_merge1}] 
      Given any \(\tta : \ta \) and \( \ttb : \tb\), we have:
      \begin{align*}
        \monMerge &(\monUnit \tta) (\monUnit \ttb) \\
                  &\equiv (\abstr{\mon \tta}{\T\ta} \abstr{\mon \ttb}{\T\tb} \abstrS{s} \case[\ta,\Ex,(\ta \times \tb)+\Ex] (\mon \tta s)  \\ 
                  &\phantomrel\equiv \quad (\abstr\tta\ta \case[\tb,\Ex,(\ta \times \tb)+\Ex] (\mon \ttb s) (\abstr\ttb\tb \inl[\ta \times \tb,\Ex] (\pair\tta\ttb)) \inr[\ta \times \tb,\Ex]) \\ 
                  &\phantomrel\equiv \quad 
        (\abstr{e_1}\Ex \case[\tb,\Ex,(\ta \times \tb)+\Ex] (\mon\ttb s) (\abstr\tta\tb \inr[\ta \times \tb,\Ex] e_1)(\abstr{e_2}\Ex \inr[\ta \times \tb,\Ex] (\exmerge e_1 e_2) )) ) \\ 
        &\phantomrel\equiv \quad 
        (\monUnit \tta) (\monUnit \ttb) \\  
        &\reducesto_\beta \abstrS s \case[\ta,\Ex,(\ta \times \tb)+\Ex] (\monUnit \tta s)  \\ 
        &\phantomrel{\reducesto_\beta} (\abstr\tta\ta \case[\tb,\Ex,(\ta \times \tb)+\Ex] (\monUnit \ttb s) (\abstr\ttb\tb \inl[\ta \times \tb,\Ex] (\pair\tta\ttb)) \inr[\ta \times \tb,\Ex]) (\dotso) \\ 
        &\equiv \abstrS s \case[\ta,\Ex,(\ta \times \tb)+\Ex] ((\abstr\tta\ta \abstrS{\_} \inl[\ta,\Ex] \tta) \tta s)  \\ 
        &\phantomrel{\reducesto_\beta} (\abstr\tta\ta \case[\tb,\Ex,(\ta \times \tb)+\Ex] (\monUnit \ttb s) (\abstr\ttb\tb \inl[\ta \times \tb,\Ex] (\pair\tta\ttb)) \inr[\ta \times \tb,\Ex]) (\dotso) \\
        &\reducesto_\beta \abstrS s \case[\ta,\Ex,(\ta \times \tb)+\Ex] (\inl[\ta,\Ex] \tta) \\ 
        &\phantomrel{\reducesto_\beta} (\abstr\tta\ta \case[\tb,\Ex,(\ta \times \tb)+\Ex] (\monUnit \ttb s) (\abstr\ttb\tb \inl[\ta \times \tb,\Ex] (\pair\tta\ttb)) \inr[\ta \times \tb,\Ex]) (\dotso) \\
        &\reducesto_+ \abstrS s \case[\tb,\Ex,(\ta \times \tb)+\Ex] (\monUnit \ttb s) (\abstr\ttb\tb \inl[\ta \times \tb,\Ex] (\pair\tta\ttb)) \inr[\ta \times \tb,\Ex] \\
        &\equiv \abstrS s \case[\tb,\Ex,(\ta \times \tb)+\Ex] ((\abstr\ttb\tb \abstrS{\_} \inl[\tb,\Ex] \ttb) \ttb s) (\abstr\ttb\tb \inl[\ta \times \tb,\Ex] (\pair \tta \ttb)) \inr[\ta \times \tb,\Ex] \\
        &\reducesto_\beta \abstrS s \case[\tb,\Ex,(\ta \times \tb)+\Ex] (\inl[\tb,\Ex] \ttb) (\abstr\ttb\tb \inl[\ta \times \tb,\Ex] (\pair \tta \ttb)) \inr[\ta \times \tb,\Ex] \\
        &\reducesto_+ \abstrS s (\abstr\ttb\tb \inl[\ta \times \tb,\Ex] (\pair\tta\ttb)) \ttb \\
        &\reducesto_\beta \abstrS s \inl[\ta \times \tb,\Ex] (\pair \tta \ttb) \\ 
        &\equiv \monUnit[\ta \times \tb] (\pair \tta \ttb), 
      \end{align*}
      as required by \cref{mon_merge1}.
  \end{itemize}
\end{proof}

\begin{proof}[Proof of \Cref{thm:ir_monadic_realizability}]
  Let \(s\) be any state. 
  We have to show that \(\monRe^s\) satisfies the properties in \cref{def:monadic_realizability_relation}. 

  \begin{description}
    \item[\ref{real1}]
      We begin with \cref{real1}, namely, for any inner interactive realizer \(\rr\) of a formula \(\fa\) with respect to \(s\), we show that: 
      \[ \monUnit \rr \monRe^s \fa. \] 
      By unfolding the definition of \(\monUnit\) we have that:
      \begin{align*}
        \monUnit \rr s &\leadsto (\abstr{\_}\State \inl \rr) s \\ 
                       &\leadsto \inl \rr, 
      \end{align*}
      thus, by definition of \(\monRe^s\), we have to check that \( \rr \re^s \fa \), which holds by assumption. 

    \item[\ref{real2}]
      In order to show \cref{real2}, for any formulas \(\fa\) and \(\fb\), we take an inner interactive realizer \(\rr\) of \(\fa \limply \fb\) with respect to \(s\), 
      that is, a term \( \rr : \m\fa \tarrow \mm\fb \) such that \( \rr \ra \) is a monadic interactive realizer of \(\fb\) with respect to \(s\), for any inner interactive realizer \(\ra\) of \(\fa\) with respect to \(s\). 
      Then we have to show that, given a monadic interactive realizer \( \mon\ra \) of \(\fa\) with respect to \(s\), we have: 
      \[ \monStar \rr \mon\ra \monRe^s \fb. \]
      By definition of \(\monRe^s\) we apply \(s\) to the realizer and 
      by unfolding the definition of \(\monStar\) and reducing we get:
      \begin{equation} \label{eq:mr2r} 
        \monStar \rr \mon\ra s \leadsto 
        \case (\mon\ra s) (\abstr\tta{\m\fa} \rr \tta s) \inr. 
      \end{equation}
      Since \(\mon\ra \monRe^s \fa \), we know that \(\mon\ra s\) reduces to either a regular value \(\inl \ra\), for some inner realizer \(\ra\) of \(\fa\) with respect to \(s\), or an exceptional value \(\inr e\), for some exception \(e\) that properly extends \(s\). 
      \begin{itemize}
        \item In the former case, \eqref{eq:mr2r} reduces to \( \rr \ra s \).
          By the assumptions we made on \(\rr\) and \(\ra\), 
          \(\rr \ra\) is a monadic interactive realizer of \(\fb\) with respect to \(s\), 
          and thus \(\rr \ra s\) reduces to either 
          a regular value which is an inner interactive realizer of \(\fb\) with respect to \(s\) or 
          an exceptional value which properly extends \(s\). 
          Thus \(\monStar \rr \mon\ra\) is a monadic interactive realizer of \(\fb\) with respect to \(s\) as required.
        \item In the latter case, \eqref{eq:mr2r} reduces to \(\inr e\). 
          Since \(e\) properly extends \(s\), \(\monStar \rr \mon\ra\) is again a monadic interactive realizer of \(\fb\) with respect to \(s\) as required.
      \end{itemize}

    \item[\ref{real3}]
      Finally we have to show \cref{real3}. 
      We assume that \(\mon\ra\) and \(\mon\rb\) are monadic interactive realizers of \(\fa\) and \(\fb\) respectively, both with respect to \(s\).
      Then we have to show that:
      \[ \monMerge \mon\ra \mon\rb \monRe^s \fa \land \fb. \] 
      By definition of \(\monRe^s\), this means we have to show that 
      \[ \monMerge \mon\ra \mon\rb s \]  
      reduces to either a regular value which is an inner interactive realizers 
      Since \(\mon\ra\) and \(\mon\rb\) are monadic interactive realizers, \(\mon\ra s\) and \(\mon\rb s\) either reduce to regular values \(\inl \ra\) and \(\inl \rb\), where \(\ra\) and \(\rb\) are inner interactive realizers of respectively \(\fa\) and \(\fb\) with respect to \(s\), or to exceptional values \(\inr e_1\) and \(\inr e_2\), where \(e_1\) and \(e_2\) properly extend \(s\).
      By unfolding the definition of \(\monMerge\) and reducing we get:
      \begin{equation} \label{eq:mr3r}
        \begin{split}
          \monMerge \mon\ra \mon\rb s 
          &\leadsto 
          \case (\mon\ra s) (\abstr\tta{\m\fa} \case (\mon\rb s) (\abstr\ttb{\m\fb} \inl (\pair \tta \ttb)) \inr) \\ 
          &\phantomrel\leadsto 
          (\abstr{e_1}\Ex \case (\mon\rb s) (\abstr\_{\m\fb} \inr e_1)(\abstr{e_2}\Ex \inr (\exmerge e_1 e_2) ))
        \end{split}
      \end{equation}
      We distinguish four cases depending on how \(\mon\ra s\) and \(\mon\rb s\) reduce:
      \begin{description}
        \item[\( \mon\ra s \leadsto \inl \ra \) and \( \mon\rb s \leadsto \inl \rb \)]
          In this case \eqref{eq:mr3r} reduces as follows:
          \begin{align*}
            \monMerge \mon\ra \mon\rb s 
            &\leadsto \case (\mon\rb s) (\abstr\ttb{\m\fb} \inl (\pair \ra \ttb)) \inr \\ 
            &\leadsto \inl (\pair \ra \rb). 
          \end{align*}
          Since it is a regular value, we have to show that \(\pair \ra \rb \re^s \fa \land \fb \). 
          This follows by definition of \(\re^s\) and from the assumption that  \(\ra \re^s \fa\) and \(\rb \re^s \fb\). 
        \item[\( \mon\ra s \leadsto \inl \ra \) and \( \mon\rb s \leadsto \inr e_2 \)]
          In this case \eqref{eq:mr3r} reduces as follows:
          \begin{align*}
            \monMerge \mon\ra \mon\rb s 
            &\leadsto \case (\mon\rb s) (\abstr\ttb{\m\fb} \inl (\pair \ra \ttb)) \inr \\ 
            &\leadsto \inr e_2. 
          \end{align*}
          Since it is an exception value, we have to show that \(e_2\) properly extends \(s\). 
          This follows by the assumption that \(\rb \re^s \fb\). 
        \item[\( \mon\ra s \leadsto \inr e_1 \) and \( \mon\rb s \leadsto \inl \rb \)]
          In this case \eqref{eq:mr3r} reduces as follows:
          \begin{align*}
            \monMerge \mon\ra \mon\rb s 
            &\leadsto \case (\mon\rb s) (\abstr\_{\m\fb} \inr e_1) (\abstr{e_2}\Ex \inr (\exmerge e_1 e_2)) \\ 
            &\leadsto \inr e_1
          \end{align*}
          Since it is an exception value, we have to show that \(e_1\) properly extends \(s\). 
          This follows by the assumption that \(\ra \re^s \fa\). 
        \item[\( \mon\ra s \leadsto \inr e_1 \) and \( \mon\rb s \leadsto \inr e_2 \)]
          \begin{align*}
            \monMerge \mon\ra \mon\rb s 
            &\leadsto \case (\mon\rb s) (\abstr\_{\m\fb} \inr e_1) (\abstr{e_2}\Ex \inr (\exmerge e_1 e_2)) \\ 
            &\leadsto \inr (\exmerge e_1 e_2)
          \end{align*}
          Since it is an exception value, we have to show that \( \exmerge e_1 e_2 \) properly extends \(s\). 
          By \cref{exmerge_prop}, this happens whenever both \(e_1\) and \(e_2\) properly extends \(s\). 
          This is the case by the assumption that \(\ra \re^s \fa\) and \(\rb \re^s \fb\). \qedhere
      \end{description}
  \end{description}
\end{proof} 

\begin{proof}[Proof of \Cref{thm:realizer_for_em}]
  Let \(\mon\rr\) and \(\fa\) stand for \(\mon\emReal(\lafa, t_1, \dotsc, t_k)\) and \(\EM(\lafa, t_1, \dotsc, t_k)\) in the following proof. 
  By \cref{def:interactive_realizability_semantics}, we have to prove that \eqref{eq:em_seq} is valid with respect to the semantics induced by \(\monRe^s\) for any given state \(s\). 

  Let the free (arithmetic) variables of \(\fa\) be \(\lva_1, \dotsc, \lva_m\) and let \(\Omega \equiv \lva_1 \coloneqq \num n_1\), \(\dotsc\), \(\lva_m \coloneqq \num n_m\) be a substitution for them. 
  Let \(\Sigma\) be a substitution for the assumption variables in \(\Gamma\). 
  Note that the only free variables in \(\mon\rr\) are arithmetic, thus \(\mon\rr[\Sigma]\) is the same as \(\mon\rr\).

  Thus we have to prove that 
  \[ \mon\rr[\Sigma,\Omega] \monRe^s \fa[\Omega]. \] 
  By definition of \(\monRe^s\), we apply \(s\) and reduce:
  \begin{align*}
    \mon\rr[\Sigma,\Omega] s \leadsto 
    \inl (\case &(\query_\lafa s t_1[\Omega] \dotsm t_k[\Omega]) \\
                &(\abstr\_\Unit \inl (\abstr\lvb\Nat \abstr\_\State \eval_\lafa t_1[\Omega] \dotsm t_k[\Omega] \lvb)) \\
                &(\abstr\lvb\Nat \inr (\pair \lvb \monUnit))), 
  \end{align*}
  and since \(\mon\rr[\Sigma,\Omega] s\) is a regular value, \(\mon\rr[\Sigma,\Omega]\) is a monadic realizer of \(\fa\) if and only if:
  \begin{equation} \label{eq:emr1}
    \begin{split}
      \case &(\query_\lafa s t_1[\Omega] \dotsm t_k[\Omega]) \\
            &(\abstr\_\Unit \inl (\abstr\lvb\Nat \abstr\_\State \eval_\lafa t_1[\Omega] \dotsm t_k[\Omega] \lvb)) \\
            &(\abstr\lvb\Nat \inr (\pair \lvb \monUnit)). 
    \end{split} 
  \end{equation}
  is an inner realizer for \(\fa\). 
  \( \query_\lafa s t_1[\Omega] \dotsm t_k[\Omega] \) reduces either to \(\inl \unit\) or to \(\inr \num n\) for some natural number \(n\). 
  We distinguish the two cases. 
  \begin{itemize}
    \item[\(\inl \unit\)]
      In the first case \eqref{eq:emr1} reduces to:
      \[
        \inl (\abstr\lvb\Nat \abstr\_\State \eval_\lafa t_1[\Omega] \dotsm t_k[\Omega] \lvb). 
      \]
      By definition of \(\re^s\), this is an inner realizer for \(\fa\) if  and only if:
      \[ 
        \rr_\forall \equiv \abstr\lvb\Nat \abstr\_\State \eval_\lafa t_1[\Omega] \dotsm t_k[\Omega] \lvb, 
      \]
      is an inner realizer for \(\qforall\lvb \lafa(t_1[\Omega], \dotsc, t_k[\Omega], \lvb)\). 
      Again by definition of \(\re^s\), this is the case if and only if 
      \[ 
        \rr_\forall \num n \monRe^s \lafa (t_1[\Omega], \dotsc, t_k[\Omega], \num n), 
      \] 
      for any natural number \(n\).
      Following the definition of \(\monRe^s\), we apply \(s\) to \(\rr_\forall \num n\) and reduce:
      \[ \rr_\forall \num n s \leadsto \eval_\lafa t_1[\Omega] \dotsm t_k[\Omega] \num n \]
      Then \( \rr_\forall \num n s \) reduces either to \(\inl \unit\) or to \(\inr e\), for some exception \(e\). 
      \begin{itemize}
        \item[\(\inl \unit\)]
          In the first case, we have to check that: 
          \[ \unit \re^s \lafa (t_1[\Omega], \dotsc, t_k[\Omega], \num n) \]
          By definition of \(\re^s\), this is the case if and only if \(\lafa (t_1[\Omega], \dotsc, t_k[\Omega], \num n) \) and this follows from \cref{eval_prop}. 
        \item[\(\inr e\)]
          In the second case, by definition of \(\monRe^s\), 
          we have to check that \(e\) properly extends \(s\) and 
          this follows from \cref{query_eval_prop}. 
      \end{itemize}
    \item[\(\inr \num n\)]
      In this case, \eqref{eq:emr1} reduces to:
      \[ \inr (\pair \num n \monUnit). \]
      By definition of \(\re^s\), this is an inner realizer for \(\fa\) if and only if 
      \[ \pair \num n \monUnit \]
      is an inner realizer for 
      \[ \qexists\lvb \lnot \lafa(t_1[\Omega], \dotsc, t_k[\Omega], \lvb). \] 
      Again by definition of \(\re^s\), this is the case if and only if 
      \[ \monUnit \re^s \lnot \lafa(t_1[\Omega], \dotsc, t_k[\Omega], \num n). \]
      Since \(\lnot \lafa(t_1[\Omega], \dotsc, t_k[\Omega], \num n) \) is defined as \( \lafa(t_1[\Omega], \dotsc, t_k[\Omega], \num n) \limply \lfalse \), 
      again by definition of \(\re^s\), we have to show that:
      \[ \monUnit u \monRe^s \lfalse, \] 
      for any inner realizer \(u\) of \( \lafa (t_1[\Omega], \dotsm, t_k[\Omega], \num n) \). 
      However, by \cref{query_prop}, \( \lafa (t_1[\Omega]\), \(\dotsm\), \(t_k[\Omega]\), \(\num n) \) does not hold, so there is no such \(u\). 
      Thus
      \[ \monUnit u \monRe^s \lfalse \] 
      holds vacuously. 
  \end{itemize}
\end{proof}

\begin{proof}[Proof of \Cref{thm:em_soundness}] 
  By definition of interactive realizability semantics, we have to prove that \(\Gamma \monSeq \mathcal{D}^* : \fa \) is valid with respect to the monadic realizability semantics induced by \(\monRe^s\) for any state \(s\). 
  So we fix a generic state \(s\) and proceed by induction on the structure of the decorated version of \(\mathcal{D}\), exactly as in \cref{thm:ha_soundness}, that is, we prove that each rule whose premisses are valid has a valid conclusion. 
  Since \(\monRe^s\) is a monadic realizability relation, this has already been shown in the proof of \cref{thm:ha_soundness} for all the rules in \(\HA\). 
  We only need to check the \(\EM\) axiom, but we have already done this in \cref{thm:realizer_for_em}. 
\end{proof}

\end{document}